\documentclass[11pt,reqno]{amsart}
\usepackage[margin=1in]{geometry}
\usepackage[T1]{fontenc}
\usepackage{amsmath,amssymb,amsthm,mathtools,mathrsfs}
\usepackage{microtype}
\usepackage{float}
\usepackage{cite}
\usepackage[hidelinks,hypertexnames=false]{hyperref}
\hypersetup{
  pdftitle={Don't mind the gap: Absolutely Continuous Edge Spectrum in a Mobility Gap},
  pdfauthor={Simon Becker and Mengxuan Yang},
  pdfsubject={Absolutely continuous edge spectrum in a mobility gap},
  pdfkeywords={absolutely continuous spectrum, Fredholm index, localization, half-space operator, random Schr\"odinger operator}}
\numberwithin{equation}{section}

\newtheorem{theorem}{Theorem}[section]
\newtheorem{proposition}[theorem]{Proposition}
\newtheorem{lemma}[theorem]{Lemma}
\newtheorem{corollary}[theorem]{Corollary}
\theoremstyle{definition}

\theoremstyle{remark}

\usepackage[nameinlink,noabbrev]{cleveref}
\usepackage{xspace}
\usepackage{tikz}
\usetikzlibrary{arrows.meta,positioning,decorations.pathreplacing}
\crefname{theorem}{Theorem}{Theorems}
\crefname{proposition}{Proposition}{Propositions}
\crefname{lemma}{Lemma}{Lemmas}
\crefname{corollary}{Corollary}{Corollaries}
\crefname{definition}{Definition}{Definitions}
\crefname{remark}{Remark}{Remarks}

\newcommand{\Hh}{\mathcal H}
\newcommand{\Kh}{\mathcal K}
\newcommand{\Eh}{\mathscr H}
\newcommand{\Aalg}{\mathcal A}
\newcommand{\Ione}{\mathcal I_1}
\newcommand{\Itwo}{\mathcal I_2}
\newcommand{\Ccal}{\mathcal C}
\newcommand{\Sone}{\mathcal S_1}
\newcommand{\Comp}{K}
\newcommand{\id}{\mathbf 1}
\newcommand{\Z}{\mathbb Z}

\newcommand{\R}{\mathbb R}
\newcommand{\C}{\mathbb C}
\newcommand{\ac}{\mathrm{ac}}

\newcommand{\ind}{\operatorname{ind}}
\newcommand{\Tr}{\operatorname{Tr}}
\newcommand{\br}[1]{\langle #1\rangle}

\newcommand{\DOne}{\hyperref[eq:discrete-locality-intro]{\textup{(D1)}}\xspace}
\newcommand{\DTwo}{\hyperref[eq:discrete-boundary-intro]{\textup{(D2)}}\xspace}
\newcommand{\DThree}{\hyperref[eq:DLoc]{\textup{(D3)}}\xspace}
\newcommand{\DFour}{\hyperref[eq:Dfinite]{\textup{(D4)}}\xspace}
\newcommand{\COne}{\hyperref[eq:CLoc1]{\textup{(C1)}}\xspace}
\newcommand{\CTwo}{\hyperref[eq:CLoc2]{\textup{(C2)}}\xspace}
\newcommand{\CThree}{\hyperref[eq:CLoc3]{\textup{(C3)}}\xspace}
\newcommand{\CFour}{\hyperref[eq:CEdge]{\textup{(C4)}}\xspace}

\title[AC spectrum in bulk localized regimes]{Don't mind the gap: \\ 
Absolutely Continuous Edge Spectrum in a Mobility Gap}

\author{Simon Becker}
\address{Bocconi University, Via Sarfatti 25, 20136 Milan, Italy}
\email{simon.becker@unibocconi.it}

\author{Mengxuan Yang}
\address{Texas A\&M University, College Station, TX 77843, USA}
\email{yangmx@tamu.edu}
\date{}

\begin{document}

\begin{abstract}
We study the edge spectrum of two-dimensional lattice Hamiltonians and magnetic Schr\"odinger operators in a mobility gap, an energy interval in which the bulk spectrum is potentially localized rather than necessarily absent.
Under suitable localization and boundary assumptions, we prove that a constant nonzero bulk topological index implies that the entire interval belongs to the absolutely continuous spectrum of the corresponding half-space operator. The absolutely continuous spectral multiplicity is bounded below by the absolute value of the index at almost every energy in the interval. The first version of the auxiliary-space cut construction in \Cref{sec:localized} was generated using ChatGPT 6.
\end{abstract}

\maketitle

\section{Introduction}
\label{sec:intro}

A two-dimensional electronic system can be insulating in the bulk while
supporting conducting states at its boundary.  In the integer quantum Hall
effect, this relation is expressed by the bulk--edge correspondence: a
topological invariant of the bulk Hamiltonian agrees with an edge index
associated with boundary currents \cite{KRS,EG,KS}.  A natural spectral
question is whether a nonzero bulk invariant forces an absolutely continuous component in the spectrum of the half-space Hamiltonian.

Disorder makes the distinction between a spectral gap and a mobility gap
important.  A spectral gap contains no bulk spectrum, whereas a mobility
gap may contain bulk eigenvalues with localized eigenfunctions.  Such
localized states are compatible with quantization of the Hall conductance
\cite{BES,EGS,GKS2}.  The question is therefore not only whether boundary
states exist, but whether an absolutely continuous component persists at
energies where the bulk states are localized.  We prove, both in the discrete and continuous settings, that it does under the hypotheses below, and that the bulk index gives a lower bound on its spectral multiplicity.

Let $H$ be a full-space self-adjoint operator on
$\Hh=\ell^2(\Z^2;\C^N)$ or $L^2(\R^2;\C^N)$, and let $\widehat H$ act on
the corresponding upper half-space, with the hypotheses in
\Cref{sec:hypotheses}.  For every Borel set $A\subset\mathbb R$, we write $\mathbf 1_{A}$ for the indicator function and $\mathbf1_A(H)$ for the
spectral projection of $H$ associated with $A$.  We introduce the multiplication
operator
\[
 \Phi(x):=\frac{x_1+ix_2}{|x_1+ix_2|}, \qquad x \ne 0,
\]
with an arbitrary value at $x=0$ (we may take $\Phi(0)=1$), and define the spectral projection
\[
 P_E:=\mathbf1_{(-\infty,E)}(H).
\]
If the bounded operator $P_E\Phi P_E+1-P_E$ is Fredholm, where $1$ denotes the identity map, we define 
\begin{equation}\label{eq:CE-intro}
 \mathcal C(E):=\ind(P_E\Phi P_E+1-P_E).
\end{equation}
This is the compression index in \cite[Definition~5.1]{ASS}; its equivalent
form on $P_E\Hh$, the Fredholm convention, and the relation to the relative
index of projections are recalled in \Cref{sec:hypotheses}.  For a
self-adjoint operator $K$, $m_{\rm ac}(E;K)$ denotes the absolutely continuous
spectral multiplicity, also defined in \Cref{sec:hypotheses}.

Absolute continuity forced by topology in a bulk spectral gap was proved in
several settings in \cite{BW,BC,DSZ}, while mobility-gap bulk--edge index
identities were established in \cite{EGS,Taarabt,BSS}.  The difference here
is that the conclusion concerns the absolutely continuous spectrum of the
physical half-space Hamiltonian while the bulk may have localized spectrum
throughout the interval; further comparisons are given under Related Works
below.

\begin{theorem}\label[theorem]{thm:main}
Let $\Delta\subset\mathbb R$ be a bounded open interval.  Assume either the
discrete hypotheses \DOne--\DFour, or the continuous hypotheses
\COne--\CFour in \Cref{sec:hypotheses}.  Suppose that the index \eqref{eq:CE-intro} is defined for every $E\in\Delta$ and is a non-zero constant 
\begin{equation}\label{eq:constant-index}
 \mathcal C(E)=\mathcal C\ne0,
 \qquad E\in\Delta.
\end{equation}
Then the multiplicity of the absolutely continuous spectrum satisfies $m_{\rm ac}(E;\widehat H)\ge |\mathcal C|$ for a.e. $E\in\Delta$,
and hence $\Delta\subset\sigma_{\rm ac}(\widehat H).$
\end{theorem}

We emphasize that in \Cref{thm:main} the interval $\Delta$ need not be a spectral gap of
$H$.  The full-space operator may have spectrum throughout $\Delta$. The
assumption is instead that this spectrum is localized in the quantitative
sense of \Cref{sec:hypotheses}.  The use of localization centers has a
precursor in \cite{BSS}, where the flux phase on localized eigenfunctions is
replaced by its value at their SULE centers.  Our new step is to introduce an
auxiliary copy of the localized spectral subspace and transfer this
center-based modification to the \emph{spatial cut}.  The physical half-space
unitary is left unregularized, while the modified cut commutes exactly with
the bulk unitary path; see \Cref{sec:localized}.

We now briefly describe the key idea of the proof. Fix $I=(a,b)\Subset J\Subset\Delta$ with
$a,b\notin\sigma_{\rm pp}(H)$ and set
\[
 P=\mathbf1_I(H),\qquad Q=1-P,\qquad
 L=\mathbf1_{(-\infty,a)}(H).
\]
Choose a localized orthonormal eigenbasis
\[
 H\psi_n=E_n\psi_n,\qquad n\in\mathcal N,
\]
of $P\mathcal H$, with centers $x_n\in\mathbb Z^2$ given by the
localization hypotheses (cf.~\Cref{sec:hypotheses}).  Let $\mathcal K=\ell^2(\mathcal N)$ with
standard basis $\{e_n\}$ and define
\[
 Ve_n=\psi_n,\qquad De_n=E_ne_n.
\]
Thus $V:\mathcal K\to P\mathcal H$ is unitary and $D$ records the energies
of the localized states.  On $\mathscr H=\mathcal H\oplus\mathcal K$
set
\[
 W=\begin{pmatrix}Q&V\\ V^*&0\end{pmatrix},
 \qquad W^*=W,\qquad W^2=1.
\]
Conjugation by $W$ replaces each localized eigenfunction $\psi_n$ by the
auxiliary basis vector $e_n$.

Let $\Lambda_1=\mathbf1_{\{x_1\ge0\}}$ and
$\Lambda_2=\mathbf1_{\{x_2\ge0\}}$ be the vertical and horizontal spatial
cut projections on $\mathcal H$.  Extend them to $\mathcal K$ by
\[
 \Lambda_{\mathcal K}e_n
 =\mathbf1_{\{(x_n)_1\ge0\}}e_n,
 \quad
 \Pi_{\mathcal K}e_n
 =\mathbf1_{\{(x_n)_2\ge0\}}e_n, \quad \text{ and set }
 \Lambda=\Lambda_1\oplus\Lambda_{\mathcal K},
 \quad
 \Pi=\Lambda_2\oplus\Pi_{\mathcal K}.
\]
The modified vertical cut is
\[
 B=W\Lambda W.
\]
Localization implies that $B-\Lambda$ is confined near $x_1=0$. Choose a smooth switch $g\in C^\infty(\mathbb R;[0,1])$ with
\[
 g=1\ \text{on }(-\infty,a],\qquad
 g=0\ \text{on }[b,\infty),\qquad g'<0\ \text{on }(a,b),
\]
and, for $0\le s\le1$, set
\[
 G_s=(1-s)(g(H)\oplus0)+s(L\oplus0),
 \qquad U_s=e^{-2\pi iG_s}.
\]
The key identity is
\[
 WG_sW=L\oplus(1-s)g(D) \text{ so }
 WU_sW=1\oplus e^{-2\pi i(1-s)g(D)}.
\]
Both $g(D)$ and $\Lambda_{\mathcal K}$ are diagonal in $\{e_n\}$; hence $[B,U_s]=0.$
Thus $B$ differs from the ordinary spatial cut only near that cut, but it
commutes exactly with the bulk unitary path.

We then compress $B$ to the enlarged half-space and round the resulting
almost-projection:
\[
 A=\Pi B\Pi,\qquad R=\mathbf1_{[1/2,1]}(A).
\]
The localization and boundary estimates imply that $[R,\widehat U_0]$ is trace-class
and an index homotopy identifies
\[
 \left|\ind(R\widehat U_0R+1-R)\right|=|\mathcal C|.
\]
Thus the bulk index is realized by the \emph{physical}, unregularized
half-space unitary.  The theorem of Asch--Bourget--Joye
\cite[Theorem~2.1(3)]{ABJ} gives a trace-class perturbation with a
bilateral-shift summand of multiplicity $|\mathcal C|$; unitary scattering
theory transfers the absolutely continuous multiplicity bound to the physical
unitary, and \Cref{lem:spectral} transfers it back to the energy variable.  The operator-theoretic argument is the same in
the discrete and continuous settings; only the localization and boundary
estimates used to verify its hypotheses differ.

The theorem neither excludes additional singular spectrum of the edge
operator nor asserts equality in the multiplicity bound.  It proves only the
lower bound forced by the bulk index.  Determining the exact absolutely
continuous multiplicity requires information beyond the argument given here.

\begin{figure}[ht]
\centering
\begin{tikzpicture}[font=\small,>=Latex]
  \begin{scope}
    \node[anchor=west,font=\bfseries] at (-3.7,4.15) {(a)};
    \fill[gray!10] (-3.2,0) rectangle (3.2,3.55);
    \draw[line width=0.9pt] (-3.2,0) -- (3.2,0);
    \draw[dashed,line width=0.7pt] (0,-0.45) -- (0,3.55);
    \node[anchor=west] at (1.45,3.20) {half-space};
    \node[anchor=west] at (0.01,0.28) {$x_1=0$};
    \node[anchor=west] at (2.05,-0.34) {$x_2=0$};

    \foreach \x/\y/\r in {-2.35/2.55/0.38,-1.35/1.20/0.34,-0.55/2.05/0.33,
                           0.75/2.70/0.36,1.55/1.10/0.34,2.35/2.05/0.37,
                           -2.25/-0.35/0.30,1.00/-0.28/0.28} {
      \draw[gray!55,line width=0.35pt] (\x,\y) circle (\r);
      \draw[gray!35,line width=0.30pt] (\x,\y) circle ({1.45*\r});
      \fill (\x,\y) circle (1.15pt);
    }
    \node[align=center,anchor=north] at (-2.02,3.58)
      {localized full-space states};
    \draw[->,thin] (-1.65,3.12) -- (-2.15,2.75);
    \node[align=center,anchor=north,font=\footnotesize] at (0,-0.72)
      {Dashed line: auxiliary cut;\\Solid line: physical boundary};
  \end{scope}

  \quad

  \begin{scope}[xshift=8.1cm]
    \node[anchor=west,font=\bfseries] at (-3.7,4.15) {(b)};
    \node[anchor=east] at (-3.05,2.70) {$H$};
    \node[anchor=east] at (-3.05,1.25) {$\widehat H$};

    \draw[->] (-2.85,2.70) -- (3.25,2.70) node[right] {$E$};
    \draw[->] (-2.85,1.25) -- (3.25,1.25) node[right] {$E$};

    \draw[densely dashed,gray!65] (-1.55,0.80) -- (-1.55,3.18);
    \draw[densely dashed,gray!65] (1.70,0.80) -- (1.70,3.18);
    \node at (0.08,3.42) {$\Delta$};
    \draw[decorate,decoration={brace,amplitude=4pt},gray!65]
       (-1.55,3.23) -- (1.70,3.23);

    \foreach \x/\h in {-1.32/0.16,-1.02/0.23,-0.70/0.13,-0.30/0.20,
                        0.02/0.14,0.38/0.22,0.82/0.15,1.12/0.19,1.46/0.13} {
      \draw[line width=0.65pt] (\x,{2.70-\h}) -- (\x,{2.70+\h});
    }
    \node[align=center,font=\footnotesize] at (0.05,2.03)
      {localized full-space spectrum may occur};

    \draw[line width=2.2pt] (-1.55,1.25) -- (1.70,1.25);
    \node[align=center,font=\footnotesize] at (0.08,0.63)
      {$m_{\rm ac}(E;\widehat H)\ge |\mathcal C|$ for a.e. $E\in\Delta$};
    \node[font=\footnotesize] at (0.08,3.78) {$\mathcal C\neq0$};
  \end{scope}
\end{tikzpicture}
\caption{Geometry and spectral conclusion of \Cref{thm:main}.
(a) The solid horizontal line is the boundary of the upper half-space and the
dashed vertical line is the auxiliary cut used in the index construction.
Dots and rings indicate localization centers and localized
full-space eigenfunctions.  (b) The same interval $\Delta$ is shown for the
full-space operator $H$ and the half-space operator $\widehat H$.
Localized full-space spectrum may occur in $\Delta$, whereas a nonzero index
$\mathcal C$ forces an absolutely continuous component of $\widehat H$ with
$m_{\rm ac}(E;\widehat H)\ge|\mathcal C|$ for almost every $E\in\Delta$.
The ticks are schematic and do not represent isolated eigenvalues; additional
singular spectrum is not excluded.}
\label{fig:geometry-spectrum}
\end{figure}
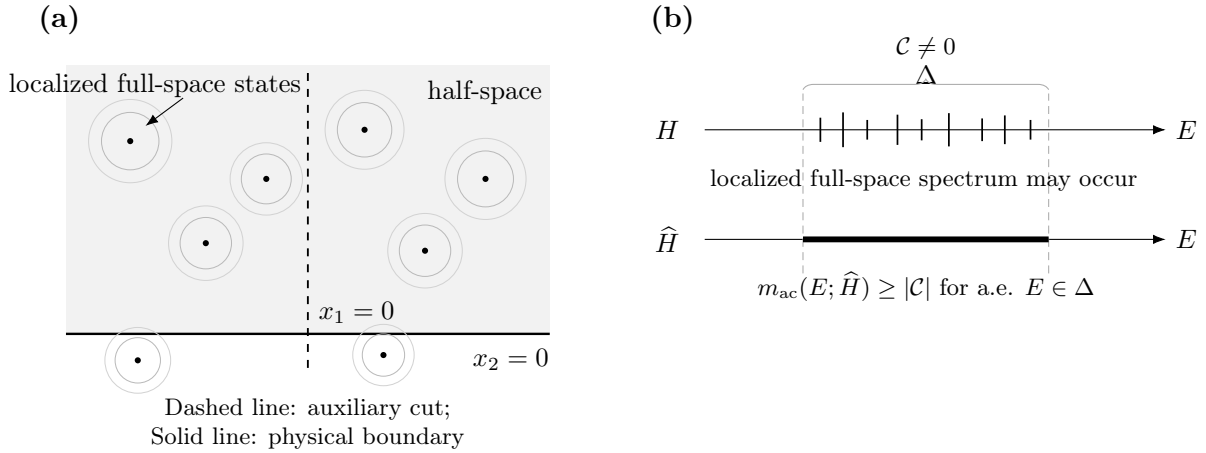

\medskip
\noindent\textbf{Related works.}
Avron--Seiler--Simon \cite{ASS} formulated charge deficiency and transport
in terms of Fredholm indices, while Bellissard--van Elst--Schulz-Baldes
\cite{BES} developed the noncommutative approach to Hall conductance and
its quantization in localized regimes.  Germinet--Klein--Schenker
\cite{GKS2} established quantization under decay assumptions on Fermi
projections for ergodic Landau Hamiltonians.

Bulk--edge correspondence in a spectral gap was proved for lattice
systems by Kellendonk--Richter--Schulz-Baldes \cite{KRS} and Elbau--Graf
\cite{EG}; Kellendonk--Schulz-Baldes \cite{KS} proved quantization of edge
currents for continuous magnetic operators.  In a mobility gap, Elgart--Graf--Schenker \cite{EGS} and Taarabt
\cite{Taarabt} proved equality of bulk and edge Hall conductances in the
discrete and continuous settings, respectively, using edge quantities
that account for localized bulk states.  Bols--Schenker--Shapiro
\cite{BSS} established bulk--edge correspondence through Fredholm
homotopies in the presence of strong disorder.  Their stability argument also
replaces the flux insertion on localized eigenfunctions by evaluation at SULE
localization centers.  Here that localization-center input is used in a
different way: the auxiliary-space dilation moves the modification to the
spatial cut, producing exact commutation with the bulk unitary path while
leaving the physical half-space unitary unchanged.  See also Stoiber
\cite{Stoiber} for stability and interface-delocalization results for strong
invariants in a mobility gap.

For spectral type, De Bi\`evre--Pul\'e \cite{DBP} and
Fr\"ohlich--Graf--Walcher \cite{FGW} proved absolute continuity for magnetic
edge Hamiltonians under suitable disorder and energy restrictions.
Bols--Werner \cite{BW} proved that a nonzero bulk index forces absolutely
continuous edge spectrum in a bulk spectral gap.  Bols--Cedzich \cite{BC}
obtained the corresponding result for topological insulators with
time-reversal symmetry squaring to $-1$, and Drouot--Shapiro--Zhu
\cite{DSZ} treated curved interfaces under suitable geometric assumptions.

The unitary index criterion of Asch--Bourget--Joye \cite{ABJ} is used in
\cite{BW} and in our proof.  The passage
from regularized edge indices to the spectrum of the physical edge
Hamiltonian in a mobility gap was raised in \cite[Section~6.2]{BSS}; see
also \cite[Section~1.6.2]{DSZ}.  Under the hypotheses stated here,
\Cref{thm:main} establishes absolutely continuous spectrum for the
unregularized half-space Hamiltonian in both the discrete and continuous
settings, with multiplicity at least $|\mathcal C|$ almost everywhere in
$\Delta$. 

\medskip
\noindent\textbf{Organization of the paper.}
\Cref{sec:proof} states the hypotheses, proves an abstract Fredholm
criterion, and gives the construction from the localized bulk spectral
subspace.  \Cref{sec:models} proves the main theorem in the discrete and
continuous settings and states the application to random models. The verification of
the random-model hypotheses is deferred to Appendix~\ref{app:random}.

\medskip
\noindent\textbf{Acknowledgement.}
M.Y. acknowledges support from the NSF grant DMS-2554813.

\medskip
\noindent\textbf{AI usage statement:} The first version of the auxiliary-space cut construction in \Cref{sec:localized} was generated using ChatGPT 6. The authors subsequently revised the construction and the manuscript with further LLM assistance to combine both discrete and continuous operators in the construction. The authors assume responsibility for the mathematical content and citations.

\section{An abstract framework}
\label{sec:proof}
\subsection{Conventions and hypotheses}
\label{sec:hypotheses}

Fix an integer $N$.  In the discrete case, the full- and half-space Hilbert spaces are
\[
 \mathcal H_{\rm d}=\ell^2(\mathbb Z^2;\mathbb C^N),
 \qquad
 \widehat{\mathcal H}_{\rm d}=\ell^2(\mathbb Z\times\mathbb N_0;\mathbb C^N),
\]
and $H,\widehat H$ are bounded self-adjoint operators.  In the continuous
case, we have
\[
 \mathcal H_{\rm c}=L^2(\mathbb R^2;\mathbb C^N),
 \qquad
 \widehat{\mathcal H}_{\rm c}=L^2(\mathbb R\times\mathbb R_+;\mathbb C^N),
\]
and $H,\widehat H$ are self-adjoint and bounded from below. For a bounded operator $T$, we say that $T$ is Fredholm if
$\operatorname{Ran}T$ is closed and
\[
 \dim\ker T+\dim\ker T^*<\infty,
 \qquad
 \ind T:=\dim\ker T-\dim\ker T^*.
\]
With respect to $\Hh=P_E\Hh\oplus(1-P_E)\Hh,$
one has
\[
 P_E\Phi P_E+1-P_E
 =\bigl(P_E\Phi P_E|_{P_E\Hh}\bigr)\oplus1,
\]
so \eqref{eq:CE-intro} is equivalently
\[
 \mathcal C(E)
 =\ind\bigl(P_E\Phi P_E:P_E\Hh\to P_E\Hh\bigr)
 =\dim\ker\bigl(P_E\Phi P_E|_{P_E\Hh}\bigr)
  -\dim\ker\bigl(P_E\Phi^*P_E|_{P_E\Hh}\bigr).
\]
This is the form used in \cite[Definition~5.1]{ASS}; see also
\cite[Proposition~2.4, Eq.~(2.12)]{ASS} for the relative-index formulation.

For a self-adjoint operator $K$, we denote by $\mathcal H_{\rm ac}(K)$
its absolutely continuous spectral subspace.  
The spectral theorem and the spectral multiplicity theorem
\cite[Theorems VII.3, VII.4, and VII.6]{RS1} give a unitary identification
\begin{equation}\label{eq:mac-definition}
  \mathcal H_{\rm ac}(K)
  \simeq
  \int_{\mathbb R}^{\oplus}
  \mathbb C^{m_{\rm ac}(E;K)}\,dE,
\end{equation}
under which $K$ acts as multiplication by the spectral parameter $(Kf)(E)=E f(E).$
Here $m_{\rm ac}(E;K)\in\{0,1,2,\ldots,\infty\}$
is defined for almost every $E$, where we denote
$\mathbb C^0=\{0\}$ and $\mathbb C^\infty=\ell^2(\mathbb N)$.  For a unitary
operator $U$ we use $m_{\rm ac}(e^{i\vartheta};U)$ analogously for the
multiplicity of its absolutely continuous part on the unit circle.

In particular, $m_{\rm ac}(E;K)\geq 1$ for a.e. $E \in I$ implies $I\subset \sigma_{\rm ac}(K).$
For an interval $I\subset\mathbb R$ and an integer $r\ge1$, let
$M_I^{(r)}$ denote multiplication on
$L^2(I;\mathbb C^r)$ by $(M_I^{(r)}f)(E)=Ef(E).$
Then $m_{\rm ac}(E;K)\ge r\ \text{for a.e. }E\in I$
if and only if there is a self-adjoint operator $K'$ such that
\[
 K|_{\mathbf1_I(K)\mathcal H_{\rm ac}(K)}
 \simeq M_I^{(r)}\oplus K'.
\]
See \cite[Chapters~VII--VIII]{RS1}.  We write $\mu_L$ for
Lebesgue measure and use the convention
\[
 \sigma_{\rm ac}(K)=
 \left\{E:\ \mu_L\!\left(
 \{t\in(E-\varepsilon,E+\varepsilon):
 m_{\rm ac}(t;K)>0\}\right)>0
 \ \text{for every }\varepsilon>0\right\}.
\]

We also write $\langle x\rangle=(1+|x|^2)^{1/2}$ and $B_R=\{x\in\mathbb R^2:|x|\le R\},$
and write $\sigma_{\rm pp}(K)=\{E\in\mathbb R:\ker(K-E)\ne\{0\}\}$
for the point spectrum of a self-adjoint operator $K$.  We write
$\mathcal S_p$ for the Schatten classes with $\|T\|_p=(\operatorname{Tr}|T|^p)^{1/p}$ (cf.~\cite[Chapter~VI]{RS1} and \cite{SimonTI}).  For $J=(a,b)$, let
$\mathcal B_1(J)$ be the set of bounded Borel functions $f$ with
$\|f\|_\infty\le1$ that are constant on $(-\infty,a]$ and on
$[b,\infty)$.  We denote these two constant values by $f_-$ and $f_+$,
respectively.

\subsubsection{Hypotheses for discrete Hamiltonians}

Let $\iota:\widehat{\mathcal H}_{\rm d}\to\mathcal H_{\rm d}$ be extension
by zero.  For $x,y\in\mathbb Z^2$, let $H_{xy}$ denote the $N\times N$
matrix block from the fiber at $y$ to the fiber at $x$; for
$x,y\in\mathbb Z\times\mathbb N_0$, define $\widehat H_{xy}$ analogously.
We assume that $H$ is local (\DOne), $\iota^*H\iota$ and $\widehat H$ differ only near the boundary (\DTwo), $H$ is localized in $\Delta$ (\DThree), and $H$ has no eigenvalues of infinite multiplicity (\DFour): 
\begin{align}
 \tag{D1}\label{eq:discrete-locality-intro}
 \|H_{xy}\|&\le Ce^{-\mu|x-y|},\\
 \tag{D2}\label{eq:discrete-boundary-intro}
 \|\widehat H_{xy}-(\iota^*H\iota)_{xy}\|
 &\le Ce^{-\mu(|x-y|+x_2+y_2)},
 \ x,y\in\mathbb Z\times\mathbb N_0, \ 
\end{align}
For every $J\Subset\Delta$ there is $\nu_J\ge0$ such that, for every
sufficiently large $m$,
\begin{equation}
 \tag{D3}\label{eq:DLoc}
 \sup_{f\in\mathcal B_1(J)}\|f(H)_{xy}\|
 \le C_{J,m}\langle x-y\rangle^{-m}\langle x\rangle^{\nu_J}.
\end{equation}
Finally, we assume 
\begin{equation}
 \tag{D4}\label{eq:Dfinite}
 \dim\ker(H-E)<\infty,
 \qquad E\in\Delta.
\end{equation}
Compare \cite[Lemma~B.2 and Appendix~B]{BSS}; see also \cite[Chapters~5--7]{AW} and \cite{GKchar}.

\subsubsection{Hypotheses for continuous Hamiltonians}

For $x\in\mathbb Z^2$, take 
\[
 C_x=x+[-\tfrac12,\tfrac12)^2,
 \qquad
 \chi_x=\mathbf1_{C_x}.
\]
We impose four continuum hypotheses: \COne gives operator-level bulk localization, \CTwo gives a localized eigenbasis, \CThree gives a uniform local trace bound, and \CFour says that the smooth functional calculi of the bulk and half-space operators differ only near the boundary.

For every $J\Subset\Delta$ there is $\nu_J\ge0$ such that, for all
sufficiently large $m$,
\begin{equation}
 \tag{C1}\label{eq:CLoc1}
 \sup_{f\in\mathcal B_1(J)}
 \|\chi_x(f(H)-f_+)\chi_y\|_2
 \le C_{J,m}\langle x-y\rangle^{-m}\langle x\rangle^{\nu_J},
\end{equation}
where $f_+$ is the value of $f$ above $J$.  Moreover
\begin{equation}
 \tag{C2}\label{eq:CLoc2}
 \mathbf1_J(H)\mathcal H_{\rm c}
 =\overline{\operatorname{span}\{\psi_n\}},
 \qquad
 H\psi_n=E_n\psi_n,
\end{equation}
with an orthonormal basis and centers $x_n\in\mathbb Z^2$ satisfying,
for every sufficiently large $m$,
\[
 \|\chi_x\psi_n\|
 \le C_{J,m}\langle x-x_n\rangle^{-m}\langle x_n\rangle^{\nu_J}.
\]
We also require the uniform local trace bound
\begin{equation}
 \tag{C3}\label{eq:CLoc3}
 \sup_{x\in\mathbb Z^2}
 \operatorname{Tr}\bigl(\chi_x\mathbf1_{(-\infty,E_0]}(H)\chi_x\bigr)
 <\infty,
 \qquad E_0\in\mathbb R.
\end{equation}
Let $\iota:\widehat{\mathcal H}_{\rm c}\to\mathcal H_{\rm c}$ be extension
by zero. For $x\in\mathbb Z\times\mathbb N_0$, write $\widehat\chi_x
 =\mathbf1_{C_x\cap(\mathbb R\times\mathbb R_+)}$
be the corresponding half-space cell cutoff.  Then for every $f\in C^\infty(\mathbb R)$ with
$\operatorname{supp}f'\Subset J$ there is an exponent $\nu_f\ge0$,
independent of the decay order, such that for every sufficiently large $m$
there is $C_{f,m}<\infty$ with
\begin{equation}
 \tag{C4}\label{eq:CEdge}
 \|\widehat\chi_x(f(\widehat H)-\iota^*f(H)\iota)\widehat\chi_y\|_1
 \le C_{f,m}\langle x-y\rangle^{-m}
      \langle x_2\rangle^{-m}\langle x\rangle^{\nu_f},
 \qquad x,y\in\mathbb Z\times\mathbb N_0.
\end{equation}
For the localization implications used below, see
\cite[Theorem~3 and Corollary~3]{GKchar}; the boundary estimate is proved for the
Anderson--Landau application in Lemma~\ref{lem:dirichlet-boundary}.

\subsection{A Fredholm criterion}\label{sec:abstract}

Let $\Eh$ be a Hilbert space.  We write 
$\mathcal B(\Eh)$ for the bounded operators on $\Eh$,
$\mathcal S_1(\Eh)$ for the trace-class operators, and
$\Comp(\Eh)$ for the compact operators (cf.~\cite{SimonTI} for the
ideal calculus used below).  

Let $\Lambda,\Pi$ be commuting orthogonal projections.  For $T\in\mathcal B(\Eh)$, we denote by $\Pi T\Pi|_{\Pi\Eh}$ its compression to $\Pi\Eh$. We identify operators on $\Pi\Eh$ with their extensions by zero to $\Eh$, and define
\[
\Lambda_+:=\Pi\Lambda\Pi\big|_{\Pi\Eh}.
\]
Let $\Aalg\subset\mathcal B(\Eh)$ be a unital $*$-algebra containing $\Lambda$ and $\Pi$, and let $\Ione,\Itwo$ be two $*$-ideals in $\Aalg$.  We assume
\begin{equation}\label{eq:ideal-assumptions}
 \Ione\Itwo+\Itwo\Ione\subset\Sone,
 \qquad [\Pi,\Ione]\subset\Sone,
 \qquad [\Lambda,\Itwo]\subset\Sone,
\end{equation}
and, for their operator-norm closures,
\begin{equation}\label{eq:ideal-closures}
 \overline{\Ione}\,\overline{\Itwo}
 +\overline{\Itwo}\,\overline{\Ione}\subset\Comp,
 \qquad [\Pi,\overline{\Ione}]\subset\Comp,
 \qquad [\Lambda,\overline{\Itwo}]\subset\Comp.
\end{equation}
In the applications, $\Ione$ and $\Itwo$ are the classes of weakly local
operators rapidly confined to the cuts $x_1=0$ and $x_2=0$, respectively.
Their precise lattice and continuum realizations are given below.

For later use, let $P=P^*=P^2$ and let $V$ be unitary.  Set
\[
 T(P,V):=PVP+1-P,\qquad F(P,V):=V^*PV-P=V^*[P,V].
\]
If $[P,V]\in\Comp$, then
\begin{align*}
 T(P,V)T(P,V^*)-1&=-PV(1-P)V^*P,\\
 T(P,V^*)T(P,V)-1&=-PV^*(1-P)VP
\end{align*}
are compact, so $T(P,V)$ is Fredholm. Write $Q:=V^*PV$.  Since $Q$ and
$1-P$ are projections,
\[
 F(P,V)+1=Q+(1-P),
\]
and therefore
\[
 \ker(F(P,V)+1)
 =\{u:Pu=u,\ Qu=0\}
 =\{u:Pu=u,\ PVu=0\}
 =\ker T(P,V).
\]
Likewise,
\[
 1-F(P,V)=(1-Q)+P,
\]
so
\[
 u\in\ker(F(P,V)-1)
 \iff Qu=u,\ Pu=0
 \iff Vu\in\ker T(P,V^*).
\]
Hence
\[
 V^*:\ker T(P,V^*)\xrightarrow{\ \simeq\ }\ker(F(P,V)-1).
\]
Consequently, we have
\begin{equation}\label{eq:index-sign}
 \ind T(P,V)=\dim\ker(F(P,V)+1)-\dim\ker(F(P,V)-1).
\end{equation}
For orthogonal projections $P,Q$ with $P-Q\in\Comp$, the relative
index (cf.~\cite[Definition~2.1]{ASS}) is
\[
 \operatorname{Ind}(P,Q) \coloneqq \dim\ker(P-Q-1)-\dim\ker(P-Q+1).
\]
Thus our convention gives
\[
 \ind T(P,V)=\operatorname{Ind}(P,V^*PV)
           =-\operatorname{Ind}(P,VPV^*).
\]
If $F(P,V)\in\mathcal S_1$, then
\[
 \ind T(P,V)=-\Tr F(P,V),
\]
by \cite[Propositions~2.2 and~2.4]{ASS} (see also 
\cite[Proposition~2.2]{ABJ}).  In particular, the flux index of
\cite{ABJ} has the opposite sign to \eqref{eq:index-sign}.

The following criterion combines the usual Fredholm homotopy invariance with
the two confinement ideals in the form needed for the modified spatial cut;
compare \cite{ASS,ABJ} and \cite{Conway,SimonTI}.

\begin{proposition}\label[proposition]{prop:abstract}
Let $B$ be an orthogonal projection on $\Eh$ such that $B-\Lambda\in\Ione.$
On $\Pi\Eh$, set
\[
 A:=\Pi B\Pi \text{ and } R:=\id_{[1/2,1]}(A).
\]
Let $U_s$, $0\le s\le1$, be a norm-continuous path of unitaries on
$\Eh$, and let $\widehat U_s$ be a norm-continuous path of unitaries on
$\Pi\Eh$.  Assume
\[
 [B,U_s]=0,\qquad U_1=1,
 \qquad D_s:=\widehat U_s-\Pi U_s\Pi\in\overline{\Itwo},
 \qquad D_0\in\Itwo.
\]
Then
\[
 [R,\widehat U_0]\in\Sone,\qquad
 \ind(R\widehat U_0R+1-R)
 =\ind(\Lambda_+\widehat U_1\Lambda_++1-\Lambda_+).
\]
\end{proposition}

\begin{proof}
Take $P_-=1-\Pi$.  Since $[\Pi,\Lambda]=0$, we have
\[
 [\Pi,B]=[\Pi,B-\Lambda]\in\Sone,\qquad
 \Pi BP_-=[\Pi,B]P_-,\qquad P_-B\Pi=-P_-[\Pi,B].
\]
It follows that
\[
 \begin{aligned}
 A-A^2
 &=\Pi B\Pi-\Pi B\Pi B\Pi
 =\Pi B(1-\Pi)B\Pi\\
 &=\Pi BP_-B\Pi
 =(P_-B\Pi)^*(P_-B\Pi)\ge0.
 \end{aligned}
\]
Since $\Pi BP_-=[\Pi,B]P_-\in\Sone$, we also have
\[
 A-A^2\in\Sone.
\]
The scalar inequality
$|\mathbf1_{[1/2,1]}(t)-t|\le2t(1-t)$ for $t\in[0,1]$ yields
\[
 S:=R-A\in\Sone,\qquad
 \|S\|_1\le2\Tr(A-A^2).
\]
Set
\[
 C:=A-\Lambda_+=\Pi(B-\Lambda)\Pi\in\Ione,
 \qquad R-\Lambda_+=C+S.
\]
The equality $BU_s=U_sB$ yields that
\[
 [A,\Pi U_s\Pi]
 =\Pi(BU_s-U_sB)\Pi-\Pi BP_-U_s\Pi+\Pi U_sP_-B\Pi
 =-\Pi BP_-U_s\Pi+\Pi U_sP_-B\Pi\in\Sone.
\]
Therefore, we have
\[
 [R,\widehat U_s]
 =[A,\Pi U_s\Pi]+[\Lambda_+,D_s]+CD_s-D_sC+[S,\widehat U_s].
\]
By \eqref{eq:ideal-assumptions} and \eqref{eq:ideal-closures}, the right
side belongs to $\Sone$ for $s=0$ and to the compact operators for every $s$. For
\[
 T_s=R\widehat U_sR+1-R,\qquad
 T_s'=R\widehat U_s^*R+1-R,
\]
we have
\[
 T_sT_s'-1=-R\widehat U_s(1-R)\widehat U_s^*R\in\Comp,
 \qquad
 T_s'T_s-1=-R\widehat U_s^*(1-R)\widehat U_sR\in\Comp.
\]
Hence every $T_s$ is Fredholm, and
\[
 \|T_s-T_t\|\le\|\widehat U_s-\widehat U_t\|\longrightarrow0
 \quad \text{as } t\to s,\qquad \ind T_0=\ind T_1.
\]
At $s=1$, take $D=\widehat U_1-1\in\overline{\Itwo}$. By compact perturbation invariance of Fredholm indices, we have
\[
 T_1-(\Lambda_+\widehat U_1\Lambda_++1-\Lambda_+)
 =RDR-\Lambda_+D\Lambda_+
 =(C+S)DR+\Lambda_+D(C+S)\in\Comp. \qedhere
\]
\end{proof}

The unitary statement behind the next lemma is the stable-current criterion of
Asch--Bourget--Joye \cite[Theorem~2.1(3)]{ABJ}, which produces $|k|$
bilateral-shift channels from the Fredholm index.  Bols--Werner \cite{BW} use
a similar idea in the spectral-gap edge problem.  The lemma records the
additional passage from the unitary spectral variable back to energy through
the strictly monotone switch function $g$.

\begin{lemma}\label[lemma]{lem:spectral}
Let $K$ be self-adjoint and let $J=(a,b)$.  Choose
$g\in C^\infty(\R;[0,1])$ such that $g=1$ on $( -\infty,a]$,
$g=0$ on $[b,\infty)$, and $g'<0$ on $J$.  Take 
$U=e^{-2\pi i g(K)}$.  If an orthogonal projection $R$ satisfies
\[
 [R,U]\in\Sone,
 \qquad k=\ind(RUR+1-R),
\]
then
\[
 m_{\mathrm{ac}}(E;K)\ge |k|
 \quad\text{for a.e. } E\in J.
\]
In particular, if $k\ne0$, then $J\subset\sigma_{\ac}(K)$.
\end{lemma}

\begin{proof}
The assertion is immediate when $k=0$.  Suppose $k\ne0$ and set $F:=U^*RU-R\in\Sone.$
By \eqref{eq:index-sign},
\[
 k:=\dim\ker(F+1)-\dim\ker(F-1).
\]
Since $k\ne0$, at least one of $\ker(F+1)$ and $\ker(F-1)$ is nonzero;
hence $1$ is an eigenvalue of $F^2$.  Because $F$ is compact, this eigenspace
has finite multiplicity.  The index convention in \cite{ABJ} is the opposite
of \eqref{eq:index-sign}, so their index is $-k$.  By
\cite[Theorem~2.1(3)]{ABJ}, there is a unitary $U_*$ and a unitary $V$ on a
complementary Hilbert space such that
\[
 U-U_*\in\Sone,\qquad U_*\simeq S_{|k|}\oplus V,
 \qquad (S_{|k|}a)_j=a_{j-1}
 \quad\text{on }\ell^2(\mathbb Z;\mathbb C^{|k|}).
\]
Let $\mathbb S^1=\{z\in\mathbb C:|z|=1\}$ and
$M_{\exp}$ denote multiplication by $e^{i\vartheta}$ on
$L^2((0,2\pi),d\vartheta/(2\pi);\mathbb C^{|k|})$
\[
 (M_{\exp}f)(\vartheta)=e^{i\vartheta}f(\vartheta).
\]
The Fourier transform gives $S_{|k|}\simeq M_{\exp},$
where $\simeq$ denotes unitary equivalence here. Since \(U-U_*\in\Sone\), the Birman--Krein theorem for unitary operators
\cite{BirmanKrein} implies that their absolutely continuous parts are
unitarily equivalent
\[
 U|_{\mathcal H_{\rm ac}(U)}
 \simeq
 U_*|_{\mathcal H_{\rm ac}(U_*)}.
\]
Since \(U_*\simeq S_{|k|}\oplus V\) and \(S_{|k|}\simeq M_{\exp}\), it follows that $m_{\rm ac}(e^{i\vartheta};U)\ge |k|$ for a.e. $\vartheta\in(0,2\pi).$
Define
\[
 \theta(E):=2\pi(1-g(E)),\qquad
 \theta'(E)=-2\pi g'(E)>0\quad(E\in J).
\]
Thus $\theta:J\to(0,2\pi)$ is a smooth increasing bijection. We denote its
inverse by $\theta^{-1}:(0,2\pi)\to J$. The spectral mapping identities are
\[
 U=e^{i\theta(K)},\qquad
 \mathbf1_{\mathbb S^1\setminus\{1\}}(U)=\mathbf1_J(K),\qquad
 \mathbf1_A(U)=\mathbf1_{\{E\in J:e^{i\theta(E)}\in A\}}(K)
\]
for every Borel set $A\subset\mathbb S^1\setminus\{1\}$.  On every compact
subinterval of $J$, both $\theta$ and its inverse are Lipschitz.  Thus, for $\mu_L$ the Lebesgue measure on $\mathbb R$
\[
 \mu_L(N)=0\ \Longrightarrow\
 \mu_L(\theta(N))=0,\qquad
 \mu_L(M)=0\ \Longrightarrow\
 \mu_L(\theta^{-1}(M))=0,
\]
for measurable null sets $N\subset J$ and $M\subset(0,2\pi)$.  Both absolutely continuous and singular spectral measures
are therefore preserved under this change of variable.

Restrict \eqref{eq:mac-definition} to $J$ and write
\[
 \mathcal H_J^{\rm ac}
 :=\mathbf1_J(K)\mathcal H_{\rm ac}(K)
 \simeq
 \int_J^\oplus \mathbb C^{m_{\rm ac}(E;K)}\,dE.
\]
After the change of variables $\vartheta=\theta(E)$, set
\[
 \widetilde{\mathcal H}_J
 :=\int_0^{2\pi}{}^\oplus
 \mathbb C^{m_{\rm ac}(\theta^{-1}(\vartheta);K)}
 \frac{d\vartheta}{2\pi}.
\]
Then
\[
 (\mathcal Jf)(\vartheta):=
 \left(\frac{2\pi}{\theta'(\theta^{-1}(\vartheta))}\right)^{1/2}
 f(\theta^{-1}(\vartheta))
\]
defines a unitary $\mathcal J:\mathcal H_J^{\rm ac}\to
\widetilde{\mathcal H}_J$, because
\[
 \int_0^{2\pi}\|(\mathcal Jf)(\vartheta)\|^2\frac{d\vartheta}{2\pi}
 =\int_J\|f(E)\|^2\,dE.
\]
On $\widetilde{\mathcal H}_J$ one has
\[
 (\mathcal J U\mathcal J^{-1}h)(\vartheta)
 =e^{i\vartheta}h(\vartheta).
\]
The orthogonal complement $\mathbf1_{\mathbb R\setminus J}(K)\Hh$ is
mapped by $U$ to the eigenvalue $1$, so it contributes no absolutely
continuous multiplicity on $\mathbb S^1\setminus\{1\}$.  Hence
\[
 m_{\rm ac}(e^{i\vartheta};U)
 =m_{\rm ac}(\theta^{-1}(\vartheta);K)
 \quad\text{for a.e. }\vartheta\in(0,2\pi).
\]
Combining this identity with the lower bound for $U$ gives
\[
 m_{\rm ac}(E;K)\ge|k|
 \quad\text{for a.e. }E\in J.
\]
For $E_0\in J$ and $\varepsilon>0$, the set
$J\cap(E_0-\varepsilon,E_0+\varepsilon)$ has positive measure. Hence, 
$E_0\in\sigma_{\rm ac}(K)$.
\end{proof}

\subsection{Construction from the localized spectral subspace}\label{sec:localized}

Fix either set of hypotheses in \Cref{thm:main},
and write $\Hh$ for the Hilbert space on which $H$ acts.  Since $\Hh$ is
separable, $H$ has at most countably many distinct eigenvalues: eigenvectors
belonging to distinct eigenvalues are orthogonal.  We may therefore choose two open
intervals
\begin{equation}\label{eq:nested-intervals}
 I=(a,b)\Subset J\Subset\Delta
\end{equation}
with $a,b$ outside the point spectrum of $H$.  The smaller interval
$I$ is the interval on which we prove the multiplicity bound; all
localization and boundary estimates are taken on the larger interval $J$.
Set
\[
 P=\id_I(H),\qquad Q=1-P,\qquad L=\id_{(-\infty,a)}(H).
\]
Let $\{\psi_n\}_{n\in \mathcal N}$ be an orthonormal eigenbasis of $P\Hh$,
$H\psi_n=E_n\psi_n$, with localization centers $x_n\in\Z^2$.  On
$\Kh=\ell^2(\mathcal N)$, let $\{e_n\}$ be the standard orthonormal basis
and define
\[
 Ve_n=\psi_n,\qquad De_n=E_ne_n.
\]
Thus $V^*V=1$, $VV^*=P$, and $HV=VD$.  On
$\Eh=\Hh\oplus\Kh$, write $W=\begin{pmatrix}Q&V\\ V^*&0\end{pmatrix}.$
The identities $QV=0$ and $V^*Q=0$ give
\[
 W^*=W,\qquad
 W^2=\begin{pmatrix}Q^2+VV^*&QV\\ V^*Q&V^*V\end{pmatrix}
     =\begin{pmatrix}Q+P&0\\0&1\end{pmatrix}=1.
\]

On $\Hh$ let $\Lambda_1$ and $\Lambda_2$ denote multiplication by
$\id_{\{x_1\ge0\}}$ and $\id_{\{x_2\ge0\}}$, respectively.  On the
physical half-space Hilbert space let $\widehat\Lambda_1$ denote
multiplication by $\id_{\{x_1\ge0\}}$.  On $\Kh$ define the corresponding
projections using the centers:
\[
 \Lambda e_n:=\id_{\{(x_n)_1\ge0\}}e_n,
 \qquad
 \Pi e_n:=\id_{\{(x_n)_2\ge0\}}e_n, \text{ and extend by linearity }
\]
and on $\Hh$ we set $\Lambda=\Lambda_1$, $\Pi=\Lambda_2$.  Thus
$\Lambda$ and $\Pi$ commute on $\Eh$.  Define $B=W\Lambda W.$  In both realizations below the localization
estimates give $W\in\Aalg$ and $[\Lambda,W]\in\Ione$.  Since \(\Ione\) is an
ideal in \(\Aalg\),
\[
 B-\Lambda=W[\Lambda,W]\in\Ione,
\]
which is the first hypothesis of Proposition~\ref{prop:abstract}.

Choose $g\in C^\infty(\mathbb R;[0,1])$ such that
\[
 g=1\ \text{on }(-\infty,a],\qquad
 g=0\ \text{on }[b,\infty),\qquad g'<0\ \text{on }I.
\]
Then $\operatorname{supp}g'\subset\overline I\Subset J$, and $g$,
$g^{1/2}$, and the spectral projection $L$ all belong to the bounded
Borel class controlled by the localization hypothesis on $J$.  Define, on
$\Eh$,
\[
 G_s=(1-s)(g(H)\oplus0)+s(L\oplus0),
 \qquad U_s=e^{-2\pi iG_s}.
\]
Since $a,b\notin\sigma_{\rm pp}(H)$,
\[
 Qg(H)Q=L,\qquad Qg(H)V=0,\qquad
 V^*g(H)V=g(D),\qquad LV=0.
\]
Hence
\[
 W(g(H)\oplus0)W=L\oplus g(D),\qquad
 W(L\oplus0)W=L\oplus0,
\]
and therefore
\[
 WG_sW=L\oplus(1-s)g(D),\qquad
 WU_sW=e^{-2\pi iL}\oplus e^{-2\pi i(1-s)g(D)}
      =1\oplus e^{-2\pi i(1-s)g(D)}.
\]
The auxiliary projection $\Lambda|_{\Kh}$ and $g(D)$ are diagonal in
$\{e_n\}$.  Thus
\[
 [B,U_s]=W[\Lambda,WU_sW]W=0,\qquad U_1=1.
\]
On the enlarged half-space $\Pi\Eh$, define
\[
 \widehat G_s=(1-s)(g(\widehat H)\oplus0)
 +s\Pi(L\oplus0)\Pi,
 \qquad \widehat U_s=e^{-2\pi i\widehat G_s}.
\]
We verify
\begin{equation}\label{eq:three-inputs}
 [\Lambda,W]\in\Ione,
 \qquad
 \widehat U_s-\Pi U_s\Pi\in\overline{\Itwo}
 \ \text{and}\ 
 \widehat U_0-\Pi U_0\Pi\in\Itwo,
\end{equation}
together with the endpoint identity
\begin{equation}\label{eq:endpoint-input}
 \left|\ind(\Lambda_+\widehat U_1\Lambda_++1-\Lambda_+)\right|
 =|\Ccal|.
\end{equation}
Assuming \eqref{eq:three-inputs}--\eqref{eq:endpoint-input},
Proposition~\ref{prop:abstract} and Lemma~\ref{lem:spectral} give the result on $I$.  Indeed,
at $s=0$,
\[
 \widehat U_0=e^{-2\pi i g(\widehat H)}\oplus1.
\]
Choose any $c\ge b$ and apply Lemma~\ref{lem:spectral} to
$K_{\rm enl}=\widehat H\oplus c\,1$.  The auxiliary summand
$c\,1$ has only point spectrum and hence no absolutely continuous part.
Therefore
\[
 m_{\rm ac}(E;K_{\rm enl})=m_{\rm ac}(E;\widehat H)
 \ge|\mathcal C|\quad\text{for a.e. }E\in I.
\]
Writing $\Delta=(\alpha,\beta)$, choose
\[
 a_j\downarrow\alpha,\qquad b_j\uparrow\beta,\qquad
 a_j,b_j\notin\sigma_{\rm pp}(H),\qquad
 I_j=(a_j,b_j)\Subset J_j\Subset\Delta.
\]
If $N_j\subset I_j$ is the exceptional null set for the preceding bound,
then $|\bigcup_jN_j|=0$ and $\bigcup_jI_j=\Delta$.  This proves the
conclusion on $\Delta$ once \eqref{eq:three-inputs} and
\eqref{eq:endpoint-input} have been verified.

\section{Discrete and continuous realizations}\label{sec:models}

\subsection{Discrete operators}

Assume \DOne--\DFour and let
$I=(a,b)\Subset J\Subset\Delta$ be as in \eqref{eq:nested-intervals}.
The estimate \eqref{eq:DLoc} on the larger interval $J$, together with
finite multiplicity, gives an orthonormal basis of
$P\mathcal H_{\mathrm d}=\mathbf 1_I(H)\mathcal H_{\mathrm d}$ consisting of eigenfunctions
$\psi_n$ with centers $x_n\in\mathbb Z^2$ and, for every sufficiently large
$m$,
\begin{equation}\label{eq:discrete-eigenfunction-decay}
 \|\psi_n(x)\|
 \le C_{J,m}\langle x-x_n\rangle^{-m}\langle x_n\rangle^{\nu_J}.
\end{equation}
Here $\nu_J$ may be increased from its value in \DThree, but remains
independent of $m$.  This is \cite[Lemma~B.2]{BSS}; for the
relation between eigenfunction-correlator decay and localized eigenfunctions,
see also \cite[Chapters~5--7]{AW}.  The centers satisfy
\begin{equation}\label{eq:center-count}
 \#\{n:|x_n|\le R\}\le C_J(1+R)^2.
\end{equation}
Indeed, if $|x_n|\le R$, then $|x|>2R$ implies
$|x-x_n|>R$.  Hence, by \eqref{eq:discrete-eigenfunction-decay}, for
$m>\nu_J+1$,
\[
 \sum_{|x|>2R}\|\psi_n(x)\|^2
 \le C_{J,m}\langle x_n\rangle^{2\nu_J}
      \sum_{|x-x_n|>R}\langle x-x_n\rangle^{-2m}
 \le C_{J,m}R^{2\nu_J+2-2m}\le\frac12
\]
for large $R$.  Since $\|\psi_n\|=1$, it follows that
$\|\mathbf1_{B_{2R}}\psi_n\|^2\ge\frac12$.  Therefore, by Bessel's
inequality applied at each $x\in B_{2R}\cap\mathbb Z^2$,
\[
 \frac12\#\{n:|x_n|\le R\}
 \le\sum_{|x_n|\le R}\|\mathbf1_{B_{2R}}\psi_n\|^2
 =\sum_{x\in B_{2R}\cap\mathbb Z^2}\sum_{|x_n|\le R}\|\psi_n(x)\|^2
 \le N\,\#(B_{2R}\cap\mathbb Z^2),
\]
which proves \eqref{eq:center-count}.

To treat the physical and auxiliary parts uniformly, introduce the label set $\mathfrak L:=\mathbb Z^2\sqcup\mathcal N.$
For a physical label \(x\in\mathbb Z^2\) set \(r_x=x\), while for an
auxiliary label \(n\in\mathcal N\) set $r_n:=x_n,$
where \(x_n\) is the localization center of \(\psi_n\).  Thus each label of
\(\mathscr H=\mathcal H\oplus\mathcal K\) is assigned a position in
\(\mathbb Z^2\).  For \(T\in\mathcal B(\mathscr H)\), let
\(T_{\alpha\beta}\) denote the block from the fiber labelled by \(\beta\)
to the fiber labelled by \(\alpha\).

We define
\[
 \mathcal A
 :=
 \left\{
 T\in\mathcal B(\mathscr H):
 \begin{array}{l}
 \text{there exists }\nu_T\ge0\text{ such that for every sufficiently large }m\\[1mm]
 \text{there exists }C_m<\infty\text{ with }
 \|T_{\alpha\beta}\|
 \le
 C_m
 \langle r_\alpha-r_\beta\rangle^{-m}
 \langle r_\alpha\rangle^{\nu_T}
 \quad\text{for all }\alpha,\beta
 \end{array}
 \right\}.
\]
For \(j=1,2\), define
\[
 \mathcal I_j
 :=
 \left\{
 T\in\mathcal A:
 \begin{array}{l}
 \text{for every sufficiently large }m\text{ there exists }C_m<\infty
 \text{ such that}\\[1mm]
 \|T_{\alpha\beta}\|
 \le
 C_m
 \langle r_\alpha-r_\beta\rangle^{-m}
 \langle (r_\alpha)_j\rangle^{-m}
 \langle r_\alpha\rangle^{\nu_T}
 \quad\text{for all }\alpha,\beta
 \end{array}
 \right\}.
\]
Here the exponent \(\nu_T\) is fixed once \(T\) is fixed and is independent
of \(m\).  The label counting satisfies
\[
 \#\{\alpha:|r_\alpha|\le R\}
 \le C(1+R)^2,
\]
by \eqref{eq:center-count} together with the corresponding bound for the
physical lattice sites.

The weak-locality calculus of Shapiro--Tauber \cite[Definition~3.2,
Lemmas~3.7--3.9, Corollary~3.11]{ST} extends to the enlarged label set because
the localization centers satisfy the polynomial counting bound
\eqref{eq:center-count}; compare also \cite{BSS}.

\begin{lemma}\label[lemma]{lem:lattice-ideals}
The lattice confinement ideals satisfy
\eqref{eq:ideal-assumptions}--\eqref{eq:ideal-closures}.  Moreover, the localized
basis gives $[\Lambda,W]\in\Ione$.
\end{lemma}

\begin{proof}
We give estimates that also apply to the continuum block decomposition.
The counting bound \eqref{eq:center-count} implies, for $r>2$,
\begin{equation}\label{eq:label-summability}
 \sum_\alpha\br{r_\alpha}^{-r}<\infty,\qquad
 \sum_\gamma\br{a-r_\gamma}^{-r}\le C_r\br a^2.
\end{equation}
Indeed,
\[
 \begin{split}
 \sum_\alpha\br{r_\alpha}^{-r}
 &\le C\sum_{j\ge0}2^{(2-r)j} \text{ and }
 \sum_\gamma\br{a-r_\gamma}^{-r}
 \le C\sum_{j\ge0}2^{-rj}(1+|a|+2^j)^2
 \le C_r\br a^2.
 \end{split}
\]
For $v\ge0$, $M>m+v+2$, and $a,b\in\mathbb Z^2$, the inequalities
\[
 \br{a-b}\le2\br{a-c}\br{c-b},\qquad
 \br c^v\le C_v\br a^v\br{a-c}^v
\]
give the convolution bound
\begin{equation}\label{eq:block-convolution}
 \begin{split}
 &\sum_\gamma\br{a-r_\gamma}^{-M}
              \br{r_\gamma-b}^{-M}\br{r_\gamma}^{v}\\
 &\qquad\le C\br{a-b}^{-m}\br a^v
       \sum_\gamma\br{a-r_\gamma}^{-M+m+v}
 \le C'\br{a-b}^{-m}\br a^{v+2}.
 \end{split}
\end{equation}
Thus $S,T\in\Aalg$ imply $ST\in\Aalg$, with the fixed weight
$\nu_{ST}=\nu_S+\nu_T+2$.  For adjoints, use
\[
 \|T^*_{\alpha\beta}\|
 =\|T_{\beta\alpha}\|
 \le C_M\br{r_\alpha-r_\beta}^{-M+\nu_T}\br{r_\alpha}^{\nu_T}.
\]
For a factor confined in direction $j$, the additional inequality
\[
 \br{a_j}^{m}\le C_m\br{a-c}^{m}\br{c_j}^{m}
\]
transfers its confinement to the output position.  Applying
\eqref{eq:block-convolution} with $M>2m+\nu_S+\nu_T+2$ proves
\[ \Aalg\mathcal I_j+\mathcal I_j\Aalg\subset\mathcal I_j\text{ and
}\mathcal I_j^*=\mathcal I_j.\]

Let $S\in\Ione$, $T\in\Itwo$ and write
$a=r_\alpha$, $b=r_\beta$, $c=r_\gamma$ and 
\[
 M(a,b,c)=\br{a-c}\br{c-b}\br{a_1}\br{c_2}.
\]
The coordinate inequalities
$|a_2|\le|a-c|+|c_2|$ and
$|c_1|\le|a-c|+|a_1|$ imply
\[
 \br a+\br b+\br c\le C M(a,b,c).
\]
Consequently, for $r>2$ and $m\ge\nu_S+\nu_T+3r$,
\[
 \begin{split}
 \|S_{\alpha\gamma}\|\,\|T_{\gamma\beta}\|
 &\le C_m M(a,b,c)^{-m}\br a^{\nu_S}\br c^{\nu_T}\le C'_r\br a^{-r}\br b^{-r}\br c^{-r}.
 \end{split}
\]
Each lattice block has rank at most $N$.  By
\eqref{eq:label-summability},
\[
 \|ST\|_1
 \le N\sum_{\alpha,\beta,\gamma}
       \|S_{\alpha\gamma}\|\,\|T_{\gamma\beta}\|
 \le C_r\left(\sum_\alpha\br{r_\alpha}^{-r}\right)^3<\infty.
\]
Interchanging coordinates proves $TS\in\Sone$.

If $S\in\Ione$, a nonzero block of $[\Pi,S]$ has
$\br{a_2}\le C\br{a-b}$.  Hence
\[
 \br a+\br b\le C\br{a-b}\br{a_1},\qquad
 \|[\Pi,S]_{\alpha\beta}\|
 \le C_r\br a^{-r}\br b^{-r}\quad(r>2).
\]
Summation proves $[\Pi,S]\in\Sone$.  Interchanging the coordinates
proves $[\Lambda,T]\in\Sone$ for $T\in\Itwo$.
If $S_n\to S$ and $T_n\to T$ in operator norm, then
\[
 \|S_nT_n-ST\|\le\|S_n-S\|\,\|T_n\|
                    +\|S\|\,\|T_n-T\|\longrightarrow0.
\]
Since $\Comp$ is norm closed, this also proves
\eqref{eq:ideal-closures}; the commutator closures follow from
$\|[P,S_n-S]\|\le2\|S_n-S\|$.

Write $\Lambda_{\Kh}=\Lambda|_{\Kh}$.  Then
\[
 [\Lambda,W]=
 \begin{pmatrix}
  -[\Lambda_1,P]&\Lambda_1V-V\Lambda_{\Kh}\\
  \Lambda_{\Kh}V^*-V^*\Lambda_1&0
 \end{pmatrix}.
\]
The upper-left block is confined by \DThree and the cut inequality.
For the upper-right block,
\[
 ((\Lambda_1V-V\Lambda_{\Kh})e_n)(x)
 =\bigl(\mathbf1_{\{x_1\ge0\}}
        -\mathbf1_{\{(x_n)_1\ge0\}}\bigr)\psi_n(x).
\]
On its support $\br{x_1}\le C\br{x-x_n}$; hence
\[
 \|((\Lambda_1V-V\Lambda_{\Kh})e_n)(x)\|
 \le C_m\br{x-x_n}^{-m}\br{x_1}^{-m}\br x^{\nu_J},
\]
by \eqref{eq:discrete-eigenfunction-decay} at order $2m+\nu_J$.
The lower-left block is its negative adjoint.  The same block estimates
without the cut factor show $W\in\Aalg$, and therefore
$[\Lambda,W]\in\Ione$.
\end{proof}

We adapt the geometric-resolvent/Helffer--Sj\"ostrand argument of
Elbau--Graf \cite[Lemma~A.3(ii)]{EG}.  Their statement is formulated with a
bulk spectral gap; under \DOne--\DTwo the resolvent argument gives the required
boundary confinement without that assumption.

\begin{lemma}\label[lemma]{lem:lattice-boundary}
Assume \DOne--\DTwo.  Let
$\widehat{\mathcal H}$ be the physical half-space Hilbert space and let
$\widehat{\mathcal I}_2\subset\mathcal B(\widehat{\mathcal H})$ be the
class of operators $T$ for which there is $\nu_T\ge0$ such that, for every
sufficiently large $m$,
\[
 \|T_{xy}\|
 \le C_m\langle x-y\rangle^{-m}
       \langle x_2\rangle^{-m}\langle x\rangle^{\nu_T},
 \qquad
 x,y\in\mathbb Z\times\mathbb N_0.
\]
If $f\in C^\infty(\mathbb R)$ is constant outside a compact interval, then
\begin{equation}\label{eq:lattice-boundary}
 f(\widehat H)-\iota^*f(H)\iota
 \in\widehat{\mathcal I}_2.
\end{equation}
In fact, for every sufficiently large $m$,
\[
 \left\|
 \bigl(f(\widehat H)-\iota^*f(H)\iota\bigr)_{xy}
 \right\|
 \le
 C_{f,m}\langle x-y\rangle^{-m}\langle x_2\rangle^{-m}.
\]
\end{lemma}

\begin{proof}
We adapt the resolvent argument in the proof of \cite[Lemma~A.3(ii)]{EG}
to the present gap-independent setting.  We define the resolvents $R(z)=(H-z)^{-1}$, $
 \widehat R(z)=(\widehat H-z)^{-1},$ and the operator
 $E_\partial=\widehat H\iota^*-\iota^*H.$
For $u\in\mathbb Z\times\mathbb N_0$ and $v\in\mathbb Z^2$,
\[
 (E_\partial)_{uv}
 =
 \begin{cases}
   \widehat H_{uv}-H_{uv}, & v_2\ge0,\\[1mm]
   -H_{uv}, & v_2<0.
 \end{cases}
\]
Hence \DOne--\DTwo and $|u-v|\ge u_2+|v_2|,$ for $ v_2<0,$
give, after decreasing $\mu>0$,
\[
 \|(E_\partial)_{uv}\|
 \le
 C e^{-\mu(|u-v|+u_2+|v_2|)}.
\]

The resolvent identity is $\widehat R(z)\iota^*-\iota^*R(z)
 =
 -\widehat R(z)E_\partial R(z).$
By \DOne--\DTwo,
\[
 \|H_{xy}\|+\|\widehat H_{xy}\|
 \le C e^{-\mu|x-y|}
\]
(on their respective lattices).  Hence, for some $\alpha>0$, the weighted
hopping sums required in the Combes--Thomas estimate are finite.  Therefore
\cite[Theorem~10.5 and Proposition~10.6]{AW}, applied also to the
half-space graph, gives, for
$\eta:=|\operatorname{Im}z|\in(0,1]$ and bounded $\operatorname{Re}z$,
the separate estimates
\[
 \|\widehat R(z)_{xu}\|
 \le C\eta^{-1}e^{-c\eta|x-u|},
 \qquad
 \|R(z)_{vy}\|
 \le C\eta^{-1}e^{-c\eta|v-y|}.
\]
Therefore, for $x,y\in\mathbb Z\times\mathbb N_0$,
\[
 \begin{aligned}
 &\left\|
 \bigl(\widehat R(z)-\iota^*R(z)\iota\bigr)_{xy}
 \right\|                                                   \le
 C\eta^{-2}
 \sum_{\substack{u_2\ge0\\v\in\mathbb Z^2}}
 e^{-c\eta|x-u|}
 e^{-\mu(|u-v|+u_2+|v_2|)}
 e^{-c\eta|v-y|}.
 \end{aligned}
\]
Since by the triangle inequality $|x-u|+|u-v|+|v-y|\ge |x-y|,$
and
\[
 x_2\le |x-u|+u_2,
 \qquad
 y_2\le |y-v|+|v_2|,
\]
the convolution estimate above yields for some $M<\infty$
\[
 \left\|
 \bigl(\widehat R(z)-\iota^*R(z)\iota\bigr)_{xy}
 \right\|
 \le
 C\eta^{-M}
 e^{-c\eta(|x-y|+x_2+y_2)}.
\]

Since $H$ and $\widehat H$ are bounded, we may replace $f$ by
$f_0\in C_c^\infty(\mathbb R)$ such that $f_0=f$ on $
 \sigma(H)\cup\sigma(\widehat H).$
Thus $f_0(H)=f(H)$ and $f_0(\widehat H)=f(\widehat H).$
Let $\widetilde f_0$ be an almost-analytic extension of order $N$
\[
 |\bar\partial\widetilde f_0(E+i\eta)|
 \le C_N|\eta|^N .
\]
The Helffer--Sj\"ostrand formula gives
\[
 \begin{aligned}
 f(\widehat H)-\iota^*f(H)\iota=
 \frac1\pi
 \int_{\mathbb C}
 \bar\partial\widetilde f_0(z)
 \bigl(
 \widehat R(z)-\iota^*R(z)\iota
 \bigr)\,d\operatorname{Re}z\,d\operatorname{Im}z .
 \end{aligned}
\]
For every $m$, one has
\[
 e^{-c\eta(|x-y|+x_2+y_2)}
 \le
 C_m\eta^{-3m}
 \langle x-y\rangle^{-m}
 \langle x_2\rangle^{-m}
 \langle y_2\rangle^{-m}.
\]
If we now choose $N>M+3m$
then
\[
 \begin{aligned}
 \left\|
 \bigl(
 f(\widehat H)-\iota^*f(H)\iota
 \bigr)_{xy}
 \right\|\le
 C_{f,m}
 \langle x-y\rangle^{-m}
 \langle x_2\rangle^{-m}
 \langle y_2\rangle^{-m}\le
 C_{f,m}
 \langle x-y\rangle^{-m}
 \langle x_2\rangle^{-m}.
 \end{aligned}
\]
Hence
\[
 f(\widehat H)-\iota^*f(H)\iota
 \in\widehat{\mathcal I}_2,
\]
which proves \eqref{eq:lattice-boundary}.
\end{proof}

We next compare the physical half-space unitary with the compressed bulk
path.  The discrepancy lies in \(\overline{\mathcal I_2}\) along the path and in
\(\mathcal I_2\) at \(s=0\); compare \cite[Lemma~3.13]{ST}.

\begin{lemma}
\label[lemma]{lem:lattice-unitary-comparison}
Assume \DOne--\DFour and use the notation of the preceding
construction.  Then
\[
 D_s:=\widehat U_s-\Pi U_s\Pi
 \in\overline{\mathcal I_2},
 \text{ for } 0\le s\le1, \text{ and }
 D_0\in\mathcal I_2.
\]
\end{lemma}

\begin{proof}
Whenever an operator in $\widehat{\mathcal I}_2$ is used below, we extend
it by zero on the auxiliary half-space and then by zero on the complement
of $\Pi\mathscr H$.  With this convention it belongs to
$\mathcal I_2\subset\mathcal B(\mathscr H)$.

Let $\Pi_-=1-\Pi$
and extend $\widehat G_s$ from $\Pi\mathscr H$ to $\mathscr H$ by $\widetilde G_s
 :=
 \widehat G_s\oplus\Pi_-G_s\Pi_-.$
Using the definitions of $G_s$ and $\widehat G_s$ gives
\begin{equation}\label{eq:decoupled-generator}
 \begin{split}
 \widetilde G_s-G_s
 &=
 (1-s)
 \bigl(g(\widehat H)-\iota^*g(H)\iota\bigr)\oplus0
 -\Pi G_s\Pi_-
 -\Pi_-G_s\Pi .
 \end{split}
\end{equation}
The first term belongs to $\mathcal I_2$ by
\Cref{lem:lattice-boundary}.  Moreover,
\[
 \Pi G_s\Pi_-=[\Pi,G_s]\Pi_-,
 \qquad
 \Pi_-G_s\Pi=-\Pi_-[\Pi,G_s].
\]
Since $G_s=F_s(H)\oplus0,
 F_s=(1-s)g+s\mathbf1_{(-\infty,a)},$
and \(F_s\in\mathcal B_1(J)\) uniformly in \(s\), \DThree gives
\[
 \|(F_s(H))_{xy}\|
 \le C_M\langle x-y\rangle^{-M}\langle x\rangle^{\nu_J}.
\]
Moreover, $[\Pi,G_s]=[\Lambda_2,F_s(H)]\oplus0,$
and $[\Lambda_2,F_s(H)]_{xy}
 =
 \bigl(\mathbf1_{\{x_2\ge0\}}-\mathbf1_{\{y_2\ge0\}}\bigr)
 (F_s(H))_{xy}.$
If this is nonzero, then \(x_2\) and \(y_2\) lie on opposite sides of the cut, so $\langle x_2\rangle\le C\langle x-y\rangle.$
Hence, using \DThree at order \(2m\),
\[
 \|[\Lambda_2,F_s(H)]_{xy}\|
 \le
 C_m\langle x-y\rangle^{-m}
 \langle x_2\rangle^{-m}
 \langle x\rangle^{\nu_J}.
\]
Thus, we have shown that $[\Pi,G_s]\in\mathcal I_2$ and consequently $\widetilde G_s-G_s\in\mathcal I_2.$

Let $q_2:\overline{\mathcal A}
 \longrightarrow
 \overline{\mathcal A}/\overline{\mathcal I_2}$ 
be the quotient map.  Since $\widetilde G_s-G_s\in\mathcal I_2$, we have $q_2(\widetilde G_s)=q_2(G_s)$ and since $q_2$ is a $*$-homomorphism, it commutes with continuous
functional calculus. Hence
\[
 q_2(\widetilde G_s)=q_2(G_s)
 \text{ which implies }
 q_2(e^{-2\pi i\widetilde G_s})
 =q_2(e^{-2\pi iG_s}).
\]
Consequently,
\[
 D_s
 =
 \Pi\bigl(
 e^{-2\pi i\widetilde G_s}-e^{-2\pi iG_s}
 \bigr)\Pi
 \in\overline{\mathcal I_2}.
\]

At $s=0$, this expression just becomes $D_0
 =
 \bigl(
 e^{-2\pi i g(\widehat H)}
 -\iota^*e^{-2\pi i g(H)}\iota
 \bigr)\oplus0.$
Applying \Cref{lem:lattice-boundary} to $f=e^{-2\pi ig}$
gives $D_0\in\mathcal I_2.$
\end{proof}

\begin{proof}[Proof of \Cref{thm:main} in the discrete case]
The localization construction above gives the first input of
Proposition~\ref{prop:abstract}, namely the projection $B$ and the unitary
path $U_s$ with
\[
 [B,U_s]=0,
 \qquad
 U_1=1.
\]
The boundary comparison required in that proposition follows from
\Cref{lem:lattice-unitary-comparison}.

It remains to compute the Fredholm index at $s=1$.  Set
\[
 L=\mathbf1_{(-\infty,a)}(H).
\]
By \cite[Lemma~4.7, Eq.~(4.11), and Lemma~4.8]{BSS},
\begin{equation}\label{eq:lattice-endpoint}
 \left|
 \ind\!\left(
 \widehat\Lambda_1
 e^{-2\pi i\iota^*L\iota}
 \widehat\Lambda_1
 +1-\widehat\Lambda_1
 \right)
 \right|
 =
 \left|
 \ind(L\Phi L+1-L)
 \right|.
\end{equation}
By \eqref{eq:constant-index}, the right-hand side equals
$|\mathcal C|$.

Since
\[
 \widehat U_1
 =
 e^{-2\pi i\Pi(L\oplus0)\Pi},
\]
the endpoint operator in Proposition~\ref{prop:abstract} is the direct sum
of the physical operator in \eqref{eq:lattice-endpoint} and the identity
on the auxiliary half-space.  The latter has Fredholm index zero.
Therefore the endpoint index has absolute value $|\mathcal C|$.

Proposition~\ref{prop:abstract} now gives
\[
 \left|
 \ind(R\widehat U_0R+1-R)
 \right|
 =
 |\mathcal C|.
\]
Finally, Lemma~\ref{lem:spectral}, applied to
\[
 \widehat U_0=e^{-2\pi ig(\widehat H)}\oplus1,
\]
implies
\[
 m_{\rm ac}(E;\widehat H)\ge|\mathcal C|
 \qquad
 \text{for a.e. }E\in I.
\]
Since $I\Subset\Delta$ was arbitrary,
\[
 m_{\rm ac}(E;\widehat H)\ge|\mathcal C|
 \qquad
 \text{for a.e. }E\in\Delta.
\]
This proves \Cref{thm:main} under \DOne--\DFour.
\end{proof}

\subsection{Continuous operators}

Assume \COne--\CFour and let
$I=(a,b)\Subset J\Subset\Delta$ be as in \eqref{eq:nested-intervals}.
We use the cell decomposition from \eqref{eq:CLoc1}--\eqref{eq:CLoc3} and,
from the localized basis on the larger interval $J$, retain those
eigenfunctions whose eigenvalues lie in $I$.  These form the basis
$\{\psi_n\}$ of $P\mathcal H_{\mathrm c}$ used in
\Cref{sec:localized}; their centers are denoted by $x_n$.  The decay in \CTwo gives, as above,
\[
 \|\mathbf1_{B_{2R}}\psi_n\|^2\ge\frac12,
 \qquad |x_n|\le R,
\]
and \CThree gives
\[
 \frac12\#\{n:|x_n|\le R\}
 \le \operatorname{Tr}\bigl(\mathbf1_{B_{2R}}P\mathbf1_{B_{2R}}\bigr)
 \le C_J(1+R)^2.
\]
Thus
\begin{equation}\label{eq:continuum-center-count}
 \#\{n:|x_n|\le R\}\le C_J(1+R)^2.
\end{equation}

On the enlarged space, use the label set
\[
 \mathfrak L:=\mathbb Z^2\sqcup\mathcal N,
 \qquad
 r_\alpha:=
 \begin{cases}
  x,   & \alpha=x\in\mathbb Z^2,\\
  x_n, & \alpha=n\in\mathcal N.
 \end{cases}
\]
The corresponding fibers are
\[
 \mathscr H_\alpha:=
 \begin{cases}
  L^2(C_x;\mathbb C^N), & \alpha=x\in\mathbb Z^2,\\
  \mathbb C e_n,        & \alpha=n\in\mathcal N,
 \end{cases}
 \qquad
 \mathscr H=\bigoplus_{\alpha\in\mathfrak L}\mathscr H_\alpha.
\]
For \(T\in\mathcal B(\mathscr H)\), write
\[
 T_{\alpha\beta}:\mathscr H_\beta\longrightarrow\mathscr H_\alpha
\]
for its \((\alpha,\beta)\)-block.

We define
\[
 \mathcal A:=
 \left\{
 T\in\mathcal B(\mathscr H):
 \begin{array}{l}
 \exists\,\nu_T\ge0\ \text{such that for every sufficiently large }m\\[1mm]
 \exists\,C_m<\infty\ \text{with}\quad
 \|T_{\alpha\beta}\|
 \le C_m
 \langle r_\alpha-r_\beta\rangle^{-m}
 \langle r_\alpha\rangle^{\nu_T}
 \quad\forall\,\alpha,\beta\in\mathfrak L
 \end{array}
 \right\}.
\]
For \(j=1,2\), let
\[
 \mathcal I_j:=
 \left\{
 T\in\mathcal A:
 \begin{array}{l}
 T_{\alpha\beta}\in
 \mathcal S_1(\mathscr H_\beta,\mathscr H_\alpha)
 \quad\forall\,\alpha,\beta,\\[1mm]
 \forall\,m\gg1\ \exists\,C_m<\infty\ \text{such that}\\[1mm]
 \|T_{\alpha\beta}\|_1
 \le C_m
 \langle r_\alpha-r_\beta\rangle^{-m}
 \langle(r_\alpha)_j\rangle^{-m}
 \langle r_\alpha\rangle^{\nu_T}
 \quad\forall\,\alpha,\beta\in\mathfrak L
 \end{array}
 \right\}.
\] 

For a physical cell $C_x$ and an auxiliary label $n$, a mixed block is an
operator between $L^2(C_x;\mathbb C^N)$ and $\mathbb C e_n$; it has rank at
most one and its trace norm equals its Hilbert-space operator norm.  An
auxiliary--auxiliary block is scalar.  These are the norms used in the same
estimate for mixed blocks.

We use the weakly-local/confined calculus of
\cite[Definition~3.2, Lemmas~3.7--3.9, Corollary~3.11]{ST} with trace-norm
decay on physical cell blocks.  This replaces finite-dimensional lattice
blocks and makes products confined in two transverse directions trace class.

\begin{lemma}\label[lemma]{lem:continuum-ideals}
The continuum confinement ideals satisfy
\eqref{eq:ideal-assumptions}--\eqref{eq:ideal-closures}, and the localized basis
in conditions \eqref{eq:CLoc1}--\eqref{eq:CLoc3} gives $[\Lambda,W]\in\Ione$.
Moreover, if an operator has arbitrary trace-norm cell decay with one fixed
polynomial weight, then its commutator with $\Lambda_j$ belongs to
$\mathcal I_j$.  In particular,
\[
 [\Lambda_2,L]\in\Itwo,\qquad [\Pi,G_s]\in\Itwo.
\]
\end{lemma}

\begin{proof}
For $P=\mathbf1_I(H)$,
\[
 P=\sum_zP\chi_zP
 \quad\text{strongly},
 \qquad
 \chi_xP\chi_y
 =\sum_z(\chi_xP\chi_z)(\chi_zP\chi_y).
\]
Schatten H\"older and \COne therefore give
\[
 \begin{split}
 \|\chi_xP\chi_y\|_1
\le\sum_z\|\chi_xP\chi_z\|_2\|\chi_zP\chi_y\|_2
 &\le C_M\br x^{\nu_J}
       \sum_z\br{x-z}^{-M}\br{z-y}^{-M}\br z^{\nu_J}\\
       &\le C_m\br{x-y}^{-m}\br x^{2\nu_J+2}.
 \end{split}
\]
The last line is \eqref{eq:block-convolution}; absolute convergence also
justifies the trace-norm sum.  The same factorization, with
$P$ replaced by $L$ or with $g(H)=g(H)^{1/2}g(H)^{1/2}$, gives
\begin{equation}\label{eq:continuum-trace-decay}
 \|\chi_xT\chi_y\|_1\le C_m\br{x-y}^{-m}\br x^{\nu_T},
 \qquad T\in\{P,L,g(H)\},
\end{equation}
with $\nu_T$ independent of $m$.  Here $g^{1/2}\in\mathcal B_1(J)$ and
its upper constant is zero, as required in \COne.
For the mixed blocks,
\[
 \|\chi_xV|_{\mathbb Ce_n}\|_1
 =\|\chi_x\psi_n\|
 \le C_m\br{x-x_n}^{-m}\br{x_n}^{\nu_J}.
\]
Thus every block of $[\Lambda,W]$ is trace class and confined in the
first coordinate, by the block formula in the proof of
\Cref{lem:lattice-ideals}.  The same estimates without the cut factor show
$W\in\mathcal A$.

For a nonzero block of $[\Lambda_j,T]$, the cells either lie on
opposite sides of the cut or at least one meets the cut.  In either case
\[
 \chi_x[\Lambda_j,T]\chi_y
 =\chi_x\Lambda_jT\chi_y-\chi_xT\Lambda_j\chi_y,
\]
and
\[
 \br{x_j}\le C\br{x-y},\qquad
 \|\chi_x[\Lambda_j,T]\chi_y\|_1
 \le2\|\chi_xT\chi_y\|_1
 \le C_m\br{x-y}^{-m}\br{x_j}^{-m}\br x^{\nu_T},
\]
using \eqref{eq:continuum-trace-decay} at order $2m$.
This proves the asserted commutator memberships.

The summability and convolution estimates
\eqref{eq:label-summability}--\eqref{eq:block-convolution} apply to the
enlarged labels by \eqref{eq:continuum-center-count}.  They show first that
$\mathcal A$ is a $*$-algebra and, using
\[
 \|(ST)_{\alpha\beta}\|_1
 \le \sum_\gamma \|S_{\alpha\gamma}\|\,
                    \|T_{\gamma\beta}\|_1,
\]
that
\[
 \mathcal A\mathcal I_j+\mathcal I_j\mathcal A\subset\mathcal I_j,
 \qquad \mathcal I_j^*=\mathcal I_j.
\]
If $S\in\Ione$ and $T\in\Itwo$, the three-position estimate from the proof
of \Cref{lem:lattice-ideals}, now with trace norm on the confined factor,
gives for every $r>2$
\[
 \sum_{\alpha,\beta,\gamma}
 \|S_{\alpha\gamma}\|\,\|T_{\gamma\beta}\|_1
 \le C_r\left(\sum_\alpha\br{r_\alpha}^{-r}\right)^3<\infty.
\]
Thus $ST\in\Sone$; interchanging the two directions gives
$TS\in\Sone$.  If $S\in\Ione$, a nonzero block of $[\Pi,S]$ either crosses
the horizontal cut or has an output cell meeting it.  In both cases the
two-position estimate gives, for $r>2$,
\[
 \|[\Pi,S]_{\alpha\beta}\|_1
 \le C_r\br{r_\alpha}^{-r}\br{r_\beta}^{-r}.
\]
Summing the block trace norms gives $[\Pi,S]\in\Sone$; the same argument
with the coordinates interchanged gives $[\Lambda,T]\in\Sone$ for
$T\in\Itwo$.  This proves \eqref{eq:ideal-assumptions}.

Finally, let $S_n\to S$ and $T_n\to T$ in operator norm with
$S_n\in\Ione$ and $T_n\in\Itwo$.  Since $S_nT_n$ and $T_nS_n$ are trace
class and converge in operator norm to $ST$ and $TS$, respectively, the
limits are compact.  Likewise,
\[
 \|[\Pi,S_n-S]\|\le2\|S_n-S\|,
 \qquad
 \|[\Lambda,T_n-T]\|\le2\|T_n-T\|,
\]
and the compact operators are norm closed.  Hence
\eqref{eq:ideal-closures} follows.
\end{proof}

In the continuum, \CFour and \Cref{lem:continuum-ideals} give the
memberships in \eqref{eq:decoupled-generator}.  Since
$\overline{\Itwo}$ is a norm-closed two-sided ideal, the two generators have
the same image in the quotient by $\overline{\Itwo}$, and continuous
functional calculus gives
\[
 \widehat U_s-\Pi U_s\Pi\in\overline{\Itwo},\qquad 0\le s\le1.
\]
At $s=0$, \CFour applied directly to $e^{-2\pi i g}$ gives the stronger
membership
\[
 \widehat U_0-\Pi U_0\Pi\in\Itwo.
\]
It remains to prove the continuum endpoint identity.  We first record a
local trace estimate used twice below.  If $L=\mathbf1_{(-\infty,a)}(H)$ and
$K\subset\mathbb R^2$ is bounded, then \CThree gives
\begin{equation}\label{eq:compact-support-compression}
 \|\mathbf1_KL\|_2^2=\Tr(\mathbf1_KL\mathbf1_K)<\infty,
 \qquad
 \|LvL\|_1\le\|v\|_\infty\|\mathbf1_KL\|_2^2
 \quad(\operatorname{supp}v\subset K).
\end{equation}

The continuum flux estimate follows the displacement splitting of
\cite[Appendix~A, Lemma~A.1]{BSS}, with local Hilbert--Schmidt cell bounds in
place of finite-dimensional lattice blocks.

\begin{lemma}\label[lemma]{lem:flux-compact}
Let $L$ be a spectral projection satisfying the estimates of conditions \eqref{eq:CLoc1}--\eqref{eq:CLoc3}.  Suppose that a bounded multiplication operator $u$ satisfies
\begin{equation}\label{eq:phase-oscillation}
 \sup_{s\in C_x,t\in C_y}|u(s)-u(t)|
 \le C\frac{\br{x-y}}{\br{x}}.
\end{equation}
Then $[L,u]\in\mathcal S_3$, in particular $[L,u]$ is compact.
\end{lemma}

\begin{proof}
Write $T_{x,b}=\chi_x[L,u]\chi_{x+b}$.  Choose points
$s_x\in C_x$ and $s_{x+b}\in C_{x+b}$.  Writing $c_x=u(s_x)$ gives
\[
 \begin{split}
 T_{x,b}
 &=\chi_xL\chi_{x+b}(u-c_{x+b})
   +(c_{x+b}-c_x)\chi_xL\chi_{x+b}-(u-c_x)\chi_xL\chi_{x+b}.
 \end{split}
\]  Since $L$ is a projection,
\[
 \|\chi_xL\chi_y\|_2^2
 \le \|\chi_xL\|_2^2
 =\operatorname{Tr}(\chi_xL\chi_x),
\]
and the last quantity is uniformly bounded by \eqref{eq:CLoc3}.  Thus
$\|\chi_xL\chi_y\|_3\le\|\chi_xL\chi_y\|_2\le C$, and
\eqref{eq:phase-oscillation} gives
\[
 \|T_{x,b}\|_3\le C\frac{\br b}{\br x}.
\]
By \eqref{eq:continuum-trace-decay}, there is a fixed exponent
$\nu_L\ge0$, independent of $m$, such that for every sufficiently large
$m$,
\[
 \|T_{x,b}\|_3\le C_m\br b^{-m}\br x^{\nu_L}.
\]
Thus
\[
 \|T_{x,b}\|_3
 \le
 \min\!\left\{
 C\frac{\br b}{\br x},
 C_m\br b^{-m}\br x^{\nu_L}
 \right\}.
\]
For fixed $b$, the blocks $T_{x,b}$ have orthogonal domains and ranges.
Splitting the sum at $\br x=\br b^\alpha$ gives
\[
 \sum_{\br x\le\br b^\alpha}\|T_{x,b}\|_3^3
 \le
 C_m\br b^{-3m+\alpha(3\nu_L+2)},
\]
while
\[
 \sum_{\br x>\br b^\alpha}\|T_{x,b}\|_3^3
 \le
 C\br b^3\sum_{\br x>\br b^\alpha}\br x^{-3}
 \le C\br b^{3-\alpha}.
\]
Thus
\begin{equation}\label{eq:S3-sum}
 \sum_x\|T_{x,b}\|_3^3
 \le C_m\br b^{-3m+\alpha(3\nu_L+2)}
      +C\br b^{3-\alpha}.
\end{equation}
Choose $\alpha=12$ and $m>12\nu_L+10$.  For fixed $b$, the blocks
$T_{x,b}$ have mutually orthogonal domains and mutually orthogonal ranges,
so
\[
 \left\|\sum_xT_{x,b}\right\|_3^3
 =\sum_x\|T_{x,b}\|_3^3.
\]
Taking cube roots in \eqref{eq:S3-sum} gives
\[
 \sum_{b\in\mathbb Z^2}
 \left\|\sum_xT_{x,b}\right\|_3
 =
 \sum_{b\in\mathbb Z^2}
 \left(\sum_x\|T_{x,b}\|_3^3\right)^{1/3}
 <\infty.
\]
Since
$[L,u]=\sum_b\sum_xT_{x,b}$ in $\mathcal S_3$, this proves
$[L,u]\in\mathcal S_3$.
\end{proof}

We will also use the following immediate variant.  If $u$ is bounded and
\eqref{eq:phase-oscillation} holds whenever $|x|$ is sufficiently large,
then
\begin{equation}\label{eq:flux-compact-outside}
 [L,u]\in\Comp.
\end{equation}
Indeed, choose a smooth cutoff $\chi$ which vanishes on a sufficiently large
ball and is one outside a slightly larger ball.  After increasing the
constant in \eqref{eq:phase-oscillation}, the multiplier $\chi u$ satisfies
that estimate globally, so \Cref{lem:flux-compact} gives
$[L,\chi u]\in\mathcal S_3$.  The multiplier $(1-\chi)u$ is compactly
supported, and \eqref{eq:compact-support-compression} gives
$[L,(1-\chi)u]\in\mathcal S_2$.

For the continuum endpoint we use the fixed-aperture phase deformation from
\cite[Lemmas~4.7--4.8]{BSS}; compactness is supplied by the preceding flux
estimate and the local trace bound.

\begin{lemma}\label[lemma]{lem:continuum-endpoint}
Let \(L=\mathbf1_{(-\infty,a)}(H)\) satisfy the cell-decay estimate obtained
from \COne and the local trace bound \CThree, and suppose
\(L\Phi L+1-L\) is Fredholm. Then
\begin{equation}\label{eq:continuum-endpoint}
 \left|\ind\!\left(
 \widehat\Lambda_1e^{-2\pi i\iota^*L\iota}\widehat\Lambda_1
 +1-\widehat\Lambda_1\right)\right|
 =
 \left|\ind(L\Phi L+1-L)\right|.
\end{equation}
\end{lemma}

\begin{proof}
Write \(\Lambda=\Lambda_1\) and, for a multiplication operator \(u\), set
\[
 \mathcal T(u):=LuL+1-L.
\]
By \eqref{eq:compact-support-compression}, \(LvL\in\Comp\) whenever
\(v\) is bounded and \(v(x)\to0\) as \(|x|\to\infty\), since
\[
 LvL=L\mathbf1_{B_R}vL+L\mathbf1_{B_R^c}vL,\qquad
 \|L\mathbf1_{B_R^c}vL\|
 \le\|\mathbf1_{B_R^c}v\|_\infty\longrightarrow0.
\]

\medskip
\noindent\emph{Step 1: moving the flux into two cones.}
Choose \(d>4\), \(0<\varepsilon<\pi/4\), and a smooth increasing
\(F:\mathbb R\to[0,1]\) with \(F=0\) on \((-\infty,-1]\) and
\(F=1\) on \([1,\infty)\). For \(x\ne(d,0)\), put
\[
 \begin{aligned}
 \theta_R(x)&=\arg(x_1-d+ix_2)\in[-\pi,\pi),
 &\qquad \varphi(\theta)&=2\pi F(\theta/\varepsilon),\\
 u_R(x)&=e^{i\varphi(\theta_R(x))},
 &\Phi_d(x)&=e^{i\theta_R(x)}.
 \end{aligned}
\]
The values at \((d,0)\) are irrelevant. Since
\[
 \Phi_d-\Phi\longrightarrow0\quad(|x|\to\infty),\qquad
 \mathcal T(\Phi_d)-\mathcal T(\Phi)=L(\Phi_d-\Phi)L\in\Comp,
\]
we have \(\ind\mathcal T(\Phi_d)=\ind\mathcal T(\Phi)\).

For \(0\le t\le1\), define
\[
 u_t(x)=\exp\!\bigl(i[(1-t)\theta_R(x)+t\varphi(\theta_R(x))]\bigr).
\]
The two boundary values across the argument cut differ by
\[
 (1-t)\pi+2\pi t-\bigl(-(1-t)\pi\bigr)=2\pi,
\]
so \(u_t\) is single-valued. Moreover,
\[
 \|u_t-u_s\|_\infty\le3\pi|t-s|,\qquad
 |\nabla u_t(x)|\le C|x-(d,0)|^{-1}.
\]
Thus \eqref{eq:phase-oscillation} holds outside a fixed bounded set and
\([L,u_t]\in\Comp\) by \eqref{eq:flux-compact-outside}. Hence
\(\mathcal T(u_t)\) is a norm-continuous Fredholm path and
\begin{equation}\label{eq:right-flux-index}
 \ind\mathcal T(\Phi)=\ind\mathcal T(u_R).
\end{equation}

Since
\[
 \operatorname{supp}(u_R-1)\subset
 C_R:=\{x_1\ge d:\ |x_2|\le\tan(\varepsilon)(x_1-d)\},
\]
we have \((1-\Lambda)u_R=1-\Lambda\), and therefore
\[
 (1-\Lambda)(\mathcal T(u_R)-1)=(1-\Lambda)[L,u_R]L\in\Comp,\qquad
 (\mathcal T(u_R)-1)(1-\Lambda)=L[u_R,L](1-\Lambda)\in\Comp.
\]
Thus
\[
 \mathcal T(u_R)-\bigl(\Lambda\mathcal T(u_R)\Lambda+1-\Lambda\bigr)
 \in\Comp,
\]
and
\[
 \ind\mathcal T(u_R)
 =\ind\bigl(\Lambda\mathcal T(u_R)\Lambda+1-\Lambda\bigr).
\]

For the left cone define
\[
 \theta_L(x)=\arg(-(x_1+d)+ix_2),\qquad
 v_L(x)=e^{i\varphi(\theta_L(x))}.
\]
Then
\[
 \operatorname{supp}(v_L-1)\subset
 C_L:=\{x_1\le-d:\ |x_2|\le\tan(\varepsilon)(-x_1-d)\},
 \qquad \Lambda(v_L-1)=0,\qquad [L,v_L]\in\Comp.
\]
Hence
\[
 \Lambda\bigl(\mathcal T(v_Lu_R)-\mathcal T(u_R)\bigr)\Lambda
 =\Lambda[L,v_L]u_RL\Lambda\in\Comp,
\]
and \eqref{eq:right-flux-index} gives
\[
 \ind\mathcal T(\Phi)
 =\ind\bigl(\Lambda\mathcal T(v_Lu_R)\Lambda+1-\Lambda\bigr).
\]

\medskip
\noindent\emph{Step 2: from the flux phase to \(L\Lambda_2L\).}
A bounded real lift of \(v_Lu_R\), away from the two vertices, is
\[
 \xi_0(x)=
 \begin{cases}
  2\pi F\!\left(\varepsilon^{-1}
       \arctan\dfrac{x_2}{x_1-d}\right),&x_1>d,\\[2pt]
  2\pi\mathbf1_{\{x_2\ge0\}},&|x_1|\le d,\\[2pt]
  2\pi F\!\left(\varepsilon^{-1}
       \arctan\dfrac{x_2}{-x_1-d}\right),&x_1<-d.
 \end{cases}
\]
Thus \(e^{i\xi_0}=v_Lu_R\). Choose
\[
 \chi\in C^\infty(\mathbb R^2),\qquad
 \chi=0\text{ on }B_{3d},\qquad \chi=1\text{ on }B_{4d}^c,
 \qquad \xi:=\chi\xi_0.
\]
The discontinuity of \(\xi_0\) on
\(\{|x_1|\le d,x_2=0\}\) lies inside \(B_{3d}\).  At the seams
\(x_1=\pm d\) outside \(B_{3d}\), one has \(|x_2|>0\); because
\(\varepsilon<\pi/4\), the outer formulas are already constant there,
equal to \(2\pi\) for \(x_2>0\) and to \(0\) for \(x_2<0\), and hence agree
smoothly with the middle branch.  Thus \(\xi_0\) is smooth wherever
\(\chi\ne0\), and
\[
 \xi\in C^\infty(\mathbb R^2;\mathbb R),\qquad
 0\le\xi\le2\pi,\qquad
 |\nabla\xi(x)|\le C\langle x\rangle^{-1},\qquad
 \operatorname{supp}(e^{i\xi}-v_Lu_R)\subset B_{4d}.
\]
By \eqref{eq:compact-support-compression},
\[
 L(e^{i\xi}-v_Lu_R)L\in\Sone\subset\Comp,
\]
and hence
\[
 \ind\mathcal T(\Phi)
 =\ind\bigl(\Lambda\mathcal T(e^{i\xi})\Lambda+1-\Lambda\bigr).
\]

Let
\[
 q:\mathcal B(\Hh)\to\mathcal B(\Hh)/\Comp(\Hh),\qquad X:=L\xi L.
\]
Since \([L,\xi]\in\Comp\), \(q(L)\) commutes with \(q(\xi)\), and therefore
\[
 q(e^{iX})
 =q(L)e^{iq(\xi)}q(L)+1-q(L)
 =q(\mathcal T(e^{i\xi})).
\]
Thus
\[
 e^{iX}-\mathcal T(e^{i\xi})\in\Comp,
\]
and
\begin{equation}\label{eq:double-flux-index}
 \ind(L\Phi L+1-L)
 =\ind\bigl(\Lambda e^{iX}\Lambda+1-\Lambda\bigr).
\end{equation}

Set
\[
 \eta:=2\pi\Lambda_2-\xi,\qquad Y:=L\eta L.
\]
Since
\[
 \operatorname{supp}\eta\subset
 \{x:\ |x_2|\le c(1+|x_1|)\},
\]
a nonzero cell block of \([\Lambda,L]\eta\) satisfies
\[
 \langle x\rangle+\langle y\rangle\le C\langle x-y\rangle.
\]
Hence, for \(r>2\),
\[
 \langle x-y\rangle^{-2r}
 \le C_r\langle x\rangle^{-r}\langle y\rangle^{-r}.
\]
Using \eqref{eq:continuum-trace-decay} at order \(m>\nu_L+2r\),
\[
 \|\chi_x[\Lambda,L]\eta\chi_y\|_1
 \le C_m\langle x-y\rangle^{-m}\langle x\rangle^{\nu_L}
 \le C_r\langle x\rangle^{-r}\langle y\rangle^{-r}.
\]
Thus
\[
 [\Lambda,L]\eta,\ \eta[\Lambda,L]\in\Sone,\qquad
 [\Lambda,Y]=[\Lambda,L]\eta L+L\eta[\Lambda,L]\in\Sone.
\]
To obtain the corresponding assertion for \(\mathcal T(e^{i\xi})\), set
\(M=e^{i\xi}-1\).  Outside \(B_{4d}\), the support of \(M\) is contained
in \(C_L\cup C_R\).  If a cell block of \([\Lambda,L]M\) or
\(M[\Lambda,L]\) is nonzero outside a fixed ball, one cell lies in one of
these cones and the two cells lie on opposite sides of the vertical cut.
Consequently,
\[
 \langle x\rangle+\langle y\rangle
 \le C\langle x-y\rangle.
\]
The compact part of the support is summable directly from
\eqref{eq:continuum-trace-decay}, while on the cones the preceding geometric
bound gives, for any \(r>2\) after taking the decay order sufficiently large,
\[
 \|\chi_x[\Lambda,L]M\chi_y\|_1
 +\|\chi_xM[\Lambda,L]\chi_y\|_1
 \le C_r\langle x\rangle^{-r}\langle y\rangle^{-r}.
\]
Hence \([\Lambda,L]M,M[\Lambda,L]\in\Sone\).  Since
\(\mathcal T(e^{i\xi})=1+LML\) and \(M\) commutes with \(\Lambda\),
\[
 [\Lambda,\mathcal T(e^{i\xi})]
 =[\Lambda,L]ML+LM[\Lambda,L]\in\Sone.
\]
Together with \(e^{iX}-\mathcal T(e^{i\xi})\in\Comp\), this also gives
\[
 [\Lambda,e^{iX}]\in\Comp.
\]
Since \(\xi\eta=\eta\xi\),
\[
 [X,Y]=-[L,\xi](1-L)\eta L-L\eta[L,\xi]L\in\Comp,
\]
so \([q(X),q(Y)]=0\).

For \(0\le t\le1\), set \(V_t=e^{i(X+tY)}\). Then
\[
 q(V_t)=q(e^{iX})e^{itq(Y)},\qquad
 [q(\Lambda),q(V_t)]=0,\qquad
 \|V_t-V_s\|\le|t-s|\,\|Y\|.
\]
Hence \(\Lambda V_t\Lambda+1-\Lambda\) is a norm-continuous Fredholm path.
Since \(X+Y=2\pi L\Lambda_2L\), \eqref{eq:double-flux-index} gives
\[
 \ind(L\Phi L+1-L)
 =\ind\bigl(\Lambda e^{2\pi iL\Lambda_2L}\Lambda+1-\Lambda\bigr).
\]
Taking adjoints,
\[
 \left|\ind(L\Phi L+1-L)\right|
 =
 \left|\ind\bigl(
 \Lambda e^{-2\pi iL\Lambda_2L}\Lambda+1-\Lambda\bigr)\right|.
\]

\medskip
\noindent\emph{Step 3: compression to the half-space.}
Set
\[
 P_2=\Lambda_2,\qquad A_0=LP_2L,\qquad
 A_1=P_2LP_2,\qquad A_t=(1-t)A_0+tA_1.
\]
Since \([L,P_2]\in\Itwo\),
\[
 \begin{aligned}
 A_0-A_1
 &=LP_2L-P_2LP_2
  =[L,P_2]L-P_2[L,P_2](1-P_2)\in\Itwo,\\
 A_0^2-A_0
 &=-LP_2(1-L)P_2L
  =-L[L,P_2](1-L)P_2L\in\Itwo,\\
 A_1^2-A_1
 &=-P_2L(1-P_2)LP_2
  =P_2[L,P_2](1-P_2)LP_2\in\Itwo.
 \end{aligned}
\]
Consequently,
\[
 A_t^2-A_t
 =(1-t)(A_0^2-A_0)+t(A_1^2-A_1)
 -t(1-t)(A_1-A_0)^2\in\Itwo.
\]

Let
\[
 q_2:\overline{\Aalg}\to\overline{\Aalg}/\overline{\Itwo}
\]
be the quotient map. Then
\[
 q_2(A_t)=q_2(A_0)=q_2(A_1),\qquad q_2(A_t)^2=q_2(A_t),
\]
and hence
\[
 q_2(e^{-2\pi iA_t})=1,\qquad
 e^{-2\pi iA_t}-1\in\overline{\Itwo},\qquad
 [\Lambda,e^{-2\pi iA_t}]\in\Comp.
\]
Thus \(\Lambda e^{-2\pi iA_t}\Lambda+1-\Lambda\) is a norm-continuous
Fredholm path.

Finally, under
\(\Hh=P_2\Hh\oplus(1-P_2)\Hh\),
\[
 A_1\simeq(\iota^*L\iota)\oplus0,\qquad
 e^{-2\pi iA_1}\simeq e^{-2\pi i\iota^*L\iota}\oplus1.
\]
Since \(\Lambda|_{P_2\Hh}=\widehat\Lambda_1\),
\[
 \Lambda e^{-2\pi iA_1}\Lambda+1-\Lambda
 \simeq
 \bigl(
 \widehat\Lambda_1e^{-2\pi i\iota^*L\iota}\widehat\Lambda_1
 +1-\widehat\Lambda_1
 \bigr)\oplus1.
\]
The second summand has index zero, which proves
\eqref{eq:continuum-endpoint}.
\end{proof}

Together with \Cref{lem:continuum-ideals} and \CFour, the preceding lemma verifies
\eqref{eq:three-inputs}--\eqref{eq:endpoint-input} in the continuum case.
The abstract argument in \Cref{sec:localized} therefore proves
\Cref{thm:main} under \COne--\CFour.

\subsection{Random operators}\label{sec:random}

Let $(\Omega_{\rm prob},\mathscr F,\mathbb P)$ be the underlying
probability space and write $\omega\in\Omega_{\rm prob}$ for a random
configuration. For a random self-adjoint operator $H_\omega$, set
\[
 P_E(\omega)=\mathbf1_{(-\infty,E)}(H_\omega).
\]

For discrete random operators we use the following standard localization
input. Let $(H_\omega,\widehat H_\omega)$ be a measurable family of full-
and half-space operators such that \DOne--\DTwo hold uniformly in
$\omega$, and assume that for every $J\Subset\Delta$,
\begin{equation}\label{eq:random-discrete-localization}
 \mathbb E\!\left[
 \sup_{f\in\mathcal B_1(J)}\|f(H_\omega)_{xy}\|
 \right]
 \le C_Je^{-\mu_J|x-y|}.
\end{equation}

\begin{corollary}\label[corollary]{cor:discrete-random}
Assume, in addition, that \DFour holds on an event of probability one
and that, on one event of probability one,
\[
 \operatorname{ind}(P_E(\omega)\Phi P_E(\omega)+1-P_E(\omega))
 =\mathcal C\ne0,
 \qquad E\in\Delta.
\]
Then, on an event of probability one,
\[
 m_{\rm ac}(E;\widehat H_\omega)\ge|\mathcal C|
 \quad\text{for a.e. }E\in\Delta.
\]
\end{corollary}

For the continuum application, consider the Anderson--Landau operator
\begin{equation}\label{eq:AL}
 H_{B,\lambda,\omega}=(-i\nabla-A_B)^2
 +\lambda\sum_{j\in\mathbb Z^2}\omega_j u(\,\cdot-j\,),
 \qquad A_B(x)=\frac B2(-x_2,x_1),
\end{equation}
where $B>0$, $\lambda\in\mathbb R$,
$u\in C_c^\infty(\mathbb R^2)$ is nonnegative and positive on a nonempty
square, and the $\omega_j$ are i.i.d. with a compactly supported law having
a bounded density. Let $\widehat H_{B,\lambda,\omega}$ be the Dirichlet
restriction to $\mathbb R\times\mathbb R_+$.

\begin{corollary}\label[corollary]{cor:AL}
Assume that an open interval $I_0\supset\overline\Delta$ lies in the
region of complete localization for $H_{B,\lambda,\omega}$ and that
\[
 \operatorname{ind}(P_E(\omega)\Phi P_E(\omega)+1-P_E(\omega))
 =\mathcal C\ne0,
 \qquad E\in\Delta,
\]
on one event of probability one. Then almost surely
\[
 m_{\rm ac}(E;\widehat H_{B,\lambda,\omega})\ge|\mathcal C|
 \quad\text{for a.e. }E\in\Delta.
\]
\end{corollary}

The verification that these random models satisfy the deterministic
hypotheses of \Cref{thm:main} is given in Appendix~\ref{app:random}.

\appendix

\section{Verification for random models}\label{app:random}

This appendix verifies the deterministic hypotheses used in the two random
applications of \Cref{sec:random}. 

\subsection{Discrete random operators}\label{app:random-discrete}

Let $\mathcal J_{\mathbb Q}(\Delta)$ be the countable set of intervals with
rational endpoints and closure in $\Delta$.  For $p>2$, set
\[
 F_J(x,y;\omega)=\sup_{f\in\mathcal B_1(J)}\|f(H_\omega)_{xy}\|,
 \qquad
 Z_{J,m}(\omega)=\sum_{x,y}\br x^{-p}\br{x-y}^{m}F_J(x,y;\omega).
\]
By Tonelli and \eqref{eq:random-discrete-localization},
\[
 \begin{split}
 \mathbb E Z_{J,m}
 &\le C_J\sum_x\br x^{-p}
          \sum_{v\in\mathbb Z^2}\br v^m e^{-\mu_J|v|}<\infty,\
 \mathbb P(\Omega_{\rm loc})=1,\quad
 \Omega_{\rm loc}:=
 \bigcap_{J\in\mathcal J_{\mathbb Q}(\Delta)}\ 
 \bigcap_{m\in\mathbb N}\{Z_{J,m}<\infty\}.
 \end{split}
\]
For $\omega\in\Omega_{\rm loc}$, each term of this nonnegative sum gives
\[
 F_J(x,y;\omega)\le Z_{J,m}(\omega)\br x^p\br{x-y}^{-m}.
\]
If $J\Subset J'\Subset\Delta$ and $J'\in\mathcal J_{\mathbb Q}(\Delta)$, then
$\mathcal B_1(J)\subset\mathcal B_1(J')$.  Thus \DThree holds
for every $J\Subset\Delta$ on the single event $\Omega_{\rm loc}$.

Thus \DThree holds on one probability-one event.  Together with the uniform
assumptions \DOne--\DTwo and the assumed \DFour, \Cref{thm:main} gives
\Cref{cor:discrete-random}; compare \cite{EGS,ST,BSS} for related
mobility-gap weak-locality arguments.

\subsection{Anderson--Landau operators}\label{app:random-continuum}
For the operator \eqref{eq:AL}, write
$V_\omega=\lambda\sum_j\omega_j u(\,\cdot-j\,)$ and choose
$M_0$ so that $|\omega_j|\le M_0$ throughout the configuration
space.  Then $\|V_\omega\|_\infty\le
 M_V:=|\lambda|M_0\sup_x\sum_j u(x-j)<\infty.$
More precisely, for the spatial domain $\mathcal O=\mathbb R^2$ or
$\mathcal O=\mathbb R\times\mathbb R_+$, the operator is associated with
\[
 \begin{split}
 q_{\mathcal O,\omega}[\psi]
 &=\|(-i\nabla-A_B)\psi\|_{L^2(\mathcal O)}^2
   +\int_\mathcal O V_\omega(x)|\psi(x)|^2\,dx,\\
 \mathcal D(q_{\mathcal O,\omega})
 &=\overline{C_c^\infty(\mathcal O)}^{\,
       (\|\psi\|_2^2+\|(-i\nabla-A_B)\psi\|_2^2)^{1/2}}.
 \end{split}
\]
These closed forms satisfy
$q_{\mathcal O,\omega}[\psi]\ge-M_V\|\psi\|_2^2$, and hence $H_{B,\lambda,\omega},\widehat H_{B,\lambda,\omega}\ge-M_V$.

Trace-ideal Combes--Thomas estimates for magnetic Schr\"odinger operators are
proved in considerable generality by Shen \cite{ShenCT}.  Closely related
magnetic bulk and Dirichlet resolvent estimates, including gauge-localized
elliptic bounds used below, appear in \cite[Appendix~A.2]{CMT}.  We record the
precise two-dimensional half-plane form needed here, including one magnetic
derivative and uniformity over the bounded disorder configurations.

\begin{lemma}\label[lemma]{lem:local-schatten}
Let $H_\omega=H_{B,\lambda,\omega}$ or its Dirichlet half-plane
restriction, and write $R_\omega(z)=(H_\omega-z)^{-1}$ and
$D_B=-i\nabla-A_B$.  Let
$K\Subset\mathbb R$ and $0<|\operatorname{Im}z|\le1$, with
$\operatorname{Re}z\in K$.  Uniformly in the bounded configuration space,
for every $p>1$ and $q>2$ there are $C,c,N>0$ such that
\begin{align*}
 \|\chi_xR_\omega(z)\chi_y\|_p
 &\le C|\operatorname{Im}z|^{-N}
 e^{-c|\operatorname{Im}z||x-y|},\\
 \|\chi_xD_BR_\omega(z)\chi_y\|_q
 &\le C|\operatorname{Im}z|^{-N}
 e^{-c|\operatorname{Im}z||x-y|}.
\end{align*}
If $\kappa>0$ is chosen so that $H_\omega+\kappa\ge1$, then at the
fixed real point $z=-\kappa$ the corresponding estimates hold with
$Ce^{-c_0|x-y|}$ on the right.
\end{lemma}

\begin{proof}
Let $\mathcal O$ denote the underlying spatial domain, either $\mathbb R^2$
or $\mathbb R\times\mathbb R_+$.  Let $C_x^*$ be a fixed bounded enlargement of $C_x$.  Choose a smooth
cutoff $\widetilde\chi_x$ supported in $C_x^*$ and equal to one on $C_x$,
and write
$\mathcal O_x=C_x^*\cap\mathcal O$.  The unitary multiplication operator $(\Gamma_x\psi)(t)=e^{iA_B(x)\cdot(t-x)}\psi(t)$
removes the constant part of the magnetic potential $A_B(t)-A_B(x)=A_B(t-x).$
It satisfies
\[
 \Gamma_x^*D_B\Gamma_x=-i\nabla-A_B(t-x),\qquad
 \sup_{t\in C_x^*}|A_B(t-x)|\le C.
\]
All derivatives of $V_\omega$ on $C_x^*$ are uniformly bounded, since
$u\in C_c^\infty$ and $|\omega_j|\le M_0$.

Take $\eta=|\operatorname{Im}z|$ and choose a smooth function $\rho_y$ with
\[
 |\nabla\rho_y|+|D^2\rho_y|\le C,
 \qquad
 |\rho_y(t)-|t-y||\le C.
\]
To justify the exponential conjugation, first replace $\rho_y$ by bounded
smooth truncations $\rho_{y,R}$ that agree with $\rho_y$ on
$\{\rho_y\le R\}$ and have the same uniform first- and second-derivative
bounds.  The magnetic Combes--Thomas estimate is uniform in $R$; passing to
the limit gives, for $\alpha=c_0\eta$ with $c_0>0$ sufficiently small,
\begin{equation}\label{eq:weighted-CT-continuum}
 \|e^{\alpha\rho_y}R_\omega(z)e^{-\alpha\rho_y}\|
 \le C\eta^{-1}.
\end{equation}
See \cite{ShenCT} and, for the bulk/half-plane magnetic estimates used here,
\cite[Appendix~A.2, Eqs.~(A.4)--(A.10)]{CMT}.  Multiplication by the gauge
$\Gamma_x$ and by the truncated weights preserves the Dirichlet trace, so
the same argument applies on the half-plane.  If
$v=R_\omega(z)\chi_yf$ and $w=e^{\alpha\rho_y}v$, then
\eqref{eq:weighted-CT-continuum} gives
\[
 \|w\|_2\le C\eta^{-1}\|f\|_2.
\]
Set
\[
 H_{\omega,\alpha,y}(z)
 :=e^{\alpha\rho_y}(H_\omega-z)e^{-\alpha\rho_y}.
\]
Then $H_{\omega,\alpha,y}(z)w=e^{\alpha\rho_y}\chi_yf$.  Conjugation by
the weight adds only first- and zeroth-order terms whose coefficients are
uniformly bounded, since $\alpha=c_0\eta\le c_0$ and
$|\nabla\rho_y|+|D^2\rho_y|\le C$.  After the local gauge transform the
principal part is therefore the ordinary Laplacian, with first- and
zeroth-order coefficients on $C_x^*$ uniformly bounded in $x$, $y$,
$\omega$, and $0<\eta\le1$.  Interior elliptic regularity, and the
corresponding Dirichlet estimate at the flat boundary, give for a cutoff
$\chi_x^*$ equal to one on $\operatorname{supp}\widetilde\chi_x$
\[
 \|\Gamma_x^*\widetilde\chi_x w\|_{H^2(\mathcal O_x)}
 \le C\bigl(
 \|\Gamma_x^*\chi_x^*H_{\omega,\alpha,y}(z)w\|_2
 +(1+|z|)\|\chi_x^*w\|_2\bigr)
 \le C\eta^{-N}\|f\|_2.
\]
The constant is uniform because the gauged coefficients are uniformly bounded
on a fixed cell enlargement and the local domains are translates of only
finitely many interior or flat-boundary geometries.
Since $e^{\alpha\rho_y}$ is bounded on $C_y$ and
$\rho_y(t)\ge |x-y|-C$ on $C_x^*$, removing the weight yields
\begin{equation}\label{eq:local-H2}
 \|\Gamma_x^*\widetilde\chi_xR_\omega(z)\chi_y\|_{L^2\to H^2(\mathcal O_x)}
 \le C\eta^{-N}e^{-c\eta|x-y|}.
\end{equation}
After enlarging the cutoff once more, the gauged magnetic derivative is a
first-order operator with uniformly bounded coefficients on $C_x^*$; hence
\begin{equation}\label{eq:local-H1}
 \|\Gamma_x^*\widetilde\chi_xD_BR_\omega(z)\chi_y\|_{L^2\to H^1(\mathcal O_x)}
 \le C\eta^{-N}e^{-c\eta|x-y|}.
\end{equation}
At $-\kappa$ the uniform resolvent bound
$\|(H_\omega+\kappa)^{-1}\|\le1$ gives the same estimates with a fixed
Combes--Thomas exponent.

The domains $\mathcal O_x$ are translates of a fixed square or one of
finitely many fixed half-squares.  Hence, for $s=1,2$, they admit Sobolev
extension operators $H^s(\mathcal O_x)\to H^s(\mathbb R^2)$ with norms
uniform in $x$.  After translation, multiply the extension by a fixed cutoff
which is one on the relevant cell and supported strictly inside
$(-\pi,\pi)^2$.  Its $2\pi$-periodization then belongs to
$H^s(\mathbb T^2)$ with a comparable norm.  Thus the singular-value
estimates reduce to
\[
 (1-\Delta_{\mathbb T^2})^{-s/2}e_k
 =(1+|k|^2)^{-s/2}e_k,\qquad k\in\mathbb Z^2,
\]
where $\{e_k\}_{k\in\mathbb Z^2}$ is the standard Fourier basis of
$L^2(\mathbb T^2)$.
Since
\[
 \#\{k\in\mathbb Z^2:|k|\le R\}\le C(1+R)^2,
\]
the embedding
$J_s:H^s(\mathcal O_x)\to L^2(C_x\cap\mathcal O)$ satisfies
\[
 s_n(J_s)\le C_s n^{-s/2},\qquad
 \|J_s\|_p^p\le C\sum_{n\ge1}n^{-sp/2}<\infty
 \quad \text{for } sp>2.
\]
Here $s_n$ denotes the $n$th singular value.  Equivalently,
\[
 J_2\in\mathcal S_p, \quad p>1;
 \qquad
 J_1\in\mathcal S_q, \quad q>2.
\]
Composing the embeddings
$J_2$ and $J_1$ with \eqref{eq:local-H2} and \eqref{eq:local-H1},
respectively, gives
\[
 \begin{split}
 \|\chi_xR_\omega(z)\chi_y\|_p
 &\le \|J_2\|_p\,
 \|\Gamma_x^*\widetilde\chi_xR_\omega(z)\chi_y\|_{L^2\to H^2(\mathcal O_x)}
 \qquad p>1,\\
 \|\chi_xD_BR_\omega(z)\chi_y\|_q
 &\le \|J_1\|_q\,
 \|\Gamma_x^*\widetilde\chi_xD_BR_\omega(z)\chi_y\|_{L^2\to H^1(\mathcal O_x)}
 \qquad q>2.
 \end{split}
\]
The gauge multipliers are unitary on $L^2$, so they do not change these
Schatten norms.  This proves the lemma.
\end{proof}

We now prove the local trace-norm boundary estimate by a geometric-resolvent
argument; compare \cite[Appendix~A]{CMT} and \cite[Appendix~A]{EG}.

\begin{lemma}\label[lemma]{lem:dirichlet-boundary}
For the model \eqref{eq:AL}, the hard Dirichlet restriction is compatible with
the full-space operator in the sense of the boundary estimate \eqref{eq:CEdge}, uniformly for
$\omega$ in the bounded configuration space.
\end{lemma}

\begin{proof}
The estimate is local and does not use spectral localization.  Abbreviate
\[
 H=H_{B,\lambda,\omega},\qquad
 \widehat H=\widehat H_{B,\lambda,\omega},\qquad
 D_B=-i\nabla-A_B,
\]
and write
\[
 R(z)=(H-z)^{-1},\qquad \widehat R(z)=(\widehat H-z)^{-1}.
\]
Let $f$ be as in \CFour.  Since $H$ and $\widehat H$ have a
common lower bound, we may choose $F\in C_c^\infty(\mathbb R)$ such that
$F=f-f_+$ on a neighborhood of
$\sigma(H)\cup\sigma(\widehat H)$.  Then
\[
 F(H)=f(H)-f_+,\qquad F(\widehat H)=f(\widehat H)-f_+,
\]
and hence
\[
 f(\widehat H)-\iota^*f(H)\iota
 =
 F(\widehat H)-\iota^*F(H)\iota.
\]
Thus it suffices to prove the estimate for $F$.

By Lemma~\ref{lem:local-schatten}, for
$\operatorname{Re}z$ in a fixed compact set and
$0<|\operatorname{Im}z|\le1$, every $p>1$ and $q>2$ satisfy
\begin{align}
 \|\chi_xR(z)\chi_y\|_p
 +\|\widehat\chi_x\widehat R(z)\widehat\chi_y\|_p
 &\le C|\operatorname{Im}z|^{-N}
 e^{-c|\operatorname{Im}z||x-y|},\label{eq:local-Sp}\\
 \|\chi_xD_BR(z)\chi_y\|_q
 +\|\widehat\chi_xD_B\widehat R(z)\widehat\chi_y\|_q\le C|\operatorname{Im}z|^{-N}
 e^{-c|\operatorname{Im}z||x-y|}.
 \label{eq:local-Sq}
\end{align}
In particular, the differentiated localized resolvent is used only with
$q>2$; no Hilbert--Schmidt assertion is made for it.

Choose $\rho=\rho(x_2)\in C^\infty(\mathbb R_+)$ with
$\rho=0$ for $x_2\le1$ and $\rho=1$ for $x_2\ge2$.  On the
support of $\nabla\rho$,
\[
 [H_{B,\lambda,\omega},\rho]
 =-2i\,\nabla\rho\cdot D_B-(\Delta\rho),
\]
so this commutator is first order and supported in a fixed boundary strip.
Because $\rho$ vanishes in a neighborhood of the boundary,
$\rho\iota^*$ maps the full-space operator domain into the Dirichlet domain.
On this domain,
\[
 \widehat H\rho\iota^*
 =\rho\iota^*H+[\widehat H,\rho]\iota^*.
\]
The commutator has bounded coefficients supported in the fixed strip, so the
identity extends after composition with the resolvent.  Consequently,
\[
 (\widehat H-z)\rho\iota^*R(z)\iota
 =\rho+[\widehat H,\rho]\iota^*R(z)\iota,
\]
and hence
\[
 \rho\iota^*R(z)\iota
 =\widehat R(z)\rho+
   \widehat R(z)[\widehat H,\rho]\iota^*R(z)\iota.
\]
Thus, if both cells lie where $\rho=1$,
\[
 \widehat\chi_x(\widehat R(z)-\iota^*R(z)\iota)\widehat\chi_y
 =-\widehat\chi_x\widehat R(z)
 [\widehat H,\rho] \,\iota^*R(z)\iota\widehat\chi_y .
\]
Let
\[
 \mathcal S=\{\zeta\in\mathbb Z\times\mathbb N_0:
 C_\zeta^*\cap\operatorname{supp}\nabla\rho\ne\varnothing\}
\]
be the cells meeting the strip.  Choose a smooth partition of unity
$\{\theta_\zeta\}_{\zeta\in\mathcal S}$ on a neighborhood of the support of
$[\widehat H,\rho]$, with $\operatorname{supp}\theta_\zeta\subset C_\zeta^*$,
and choose $\widetilde\theta_\zeta$ supported in a fixed larger cell
neighborhood such that $\widetilde\theta_\zeta=1$ on
$\operatorname{supp}\theta_\zeta$.  The family has uniformly finite overlap
and uniformly bounded derivatives.  Since the coefficients of
$[\widehat H,\rho]$ are supported in this strip,
\[
 [\widehat H,\rho]
 =\sum_{\zeta\in\mathcal S}\theta_\zeta[\widehat H,\rho],
\]
and therefore
\[
 \widehat\chi_x\widehat R(z)[\widehat H,\rho]\iota^*R(z)\iota\widehat\chi_y
 =\sum_{\zeta\in\mathcal S}
 \widehat\chi_x\widehat R(z)\widetilde\theta_\zeta\,
 \theta_\zeta[\widehat H,\rho]\iota^*R(z)\iota\widehat\chi_y.
\]
The smooth cutoffs are supported in uniformly many cells, so
\eqref{eq:local-Sp}--\eqref{eq:local-Sq} apply to them with the same type of
constants.  With $p=3/2$ and $q=3$, the first-order part of the $\zeta$ term
is bounded by
\[
 C\,\|\widehat\chi_x\widehat R(z)\widetilde\theta_\zeta\|_{3/2}
 \|\theta_\zeta D_BR(z)\iota\widehat\chi_y\|_3,
\]
and its zeroth-order part by
\[
 C\,\|\widehat\chi_x\widehat R(z)\widetilde\theta_\zeta\|_{3/2}
 \|\theta_\zeta R(z)\iota\widehat\chi_y\|_3.
\]
Here
\[
 \frac1{3/2}+\frac13=1,
 \qquad
 \|AB\|_1\le\|A\|_{3/2}\|B\|_3.
\]
Thus Schatten H\"older and \eqref{eq:local-Sp}--\eqref{eq:local-Sq} give
\[
 \|\widehat\chi_x\widehat R(z)
   [\widehat H,\rho]\iota^*R(z)\iota\widehat\chi_y\|_1
 \le C|\operatorname{Im}z|^{-N}
 \sum_{\zeta\in\mathcal S}
 e^{-c|\operatorname{Im}z|(|x-\zeta|+|\zeta-y|)}.
\]
Take $t=|\operatorname{Im}z|$.  For strip cells $\zeta$,
\[
 |x-\zeta|+|\zeta-y|\ge|x-y|,\qquad
 |x-\zeta|+|\zeta-y|\ge x_2+y_2-C.
\]
Hence, after decreasing $c$,
\[
 e^{-ct(|x-\zeta|+|\zeta-y|)}
 \le
 e^{-c't(|x-y|+x_2+y_2)}
 e^{-c't(|x-\zeta|+|\zeta-y|)}.
\]
Splitting off part of the exponential and summing the remaining
one-dimensional exponential gives
\[
 \sum_{n\in\mathbb Z}e^{-c't|n|}
 \le 1+2\int_0^\infty e^{-c'ts}\,ds
 \le Ct^{-1},
\]
and therefore
\[
 \sum_{\zeta\in\mathcal S}
 e^{-ct(|x-\zeta|+|\zeta-y|)}
 \le C't^{-1}e^{-c't(|x-y|+x_2+y_2)}.
\]
Consequently
\begin{equation}\label{eq:far-resolvent-difference}
 \|\widehat\chi_x(\widehat R(z)-\iota^*R(z)\iota)\widehat\chi_y\|_1
 \le C|\operatorname{Im}z|^{-N'}
 e^{-c'|\operatorname{Im}z|(|x-y|+x_2+y_2)}
\end{equation}
whenever both cells lie outside the fixed boundary strip.

It remains to treat cells meeting that strip.  We record explicitly the
trace-norm estimate for smooth functional calculus.  Choose $\kappa>0$ so that
$H+\kappa\ge1$ and $\widehat H+\kappa\ge1$, and write
\[
 G_2(t)=(t+\kappa)^2G(t),\qquad G\in C_c^\infty(\mathbb R).
\]
Then
\[
 G(H)=(H+\kappa)^{-2}G_2(H),
 \qquad
 G(\widehat H)=(\widehat H+\kappa)^{-2}G_2(\widehat H).
\]
At the fixed spectral parameter $-\kappa$,
Lemma~\ref{lem:local-schatten} with $p=2$ and a unit-cell partition gives
\[
 \|\chi_x(H+\kappa)^{-2}\chi_z\|_1
 +\|\widehat\chi_x(\widehat H+\kappa)^{-2}\widehat\chi_z\|_1
 \le Ce^{-c_0|x-z|}.
\]
For example, with
\[
 R_\kappa=(H+\kappa)^{-1},
 \qquad
 R_\kappa^2=(H+\kappa)^{-2},
\]
we have
\[
 \|\chi_xR_\kappa^2\chi_z\|_1
 \le\sum_w\|\chi_xR_\kappa\chi_w\|_2\|\chi_wR_\kappa\chi_z\|_2
 \le C\sum_we^{-c_0(|x-w|+|w-z|)}
 \le C'e^{-c_0'|x-z|}.
\]
Helffer--Sj\"ostrand functional calculus combined with the Combes--Thomas
estimate gives, for every $M$ (see also \cite{ShenCT}),
\[
 \|\chi_zG_2(H)\chi_y\|
 +\|\widehat\chi_zG_2(\widehat H)\widehat\chi_y\|
 \le C_M\langle z-y\rangle^{-M}.
\]
After summing over $z$, Schatten H\"older yields
\begin{equation}\label{eq:smooth-offdiag-trace}
 \|\widehat\chi_xG(\widehat H)\widehat\chi_y\|_1
 +\|\widehat\chi_x\iota^*G(H)\iota\widehat\chi_y\|_1
 \le C_M\langle x-y\rangle^{-M}.
\end{equation}
Thus \eqref{eq:CEdge} follows immediately if the $x$-cell meets the
boundary strip, because there
\[
 \br{x_2}^{-m}\asymp1.
\]
If only the $y$-cell meets the strip, then
\[
 \br{x_2}\le C\br{x-y},
\]
and applying
\eqref{eq:smooth-offdiag-trace} with order $2m$ gives the required factor
$\langle x-y\rangle^{-m}\langle x_2\rangle^{-m}$.

Choose $\vartheta\in C_c^\infty((-1,1))$ equal to one near zero and
an integer $M$.  An almost-analytic extension is
\[
 \widetilde F_M(s+it)=\vartheta(t)
       \sum_{j=0}^M\frac{F^{(j)}(s)(it)^j}{j!},\qquad
 \overline\partial=\tfrac12(\partial_s+i\partial_t).
\]
The Taylor terms cancel in $\overline\partial\widetilde F_M$ up to
order $M$, so
\[
 |\overline\partial\widetilde F_M(s+it)|
 \le C_M|t|^M.
\]
With our convention $R(z)=(H-z)^{-1}$, the functional calculus formula is
\[
 F(\widehat H)-\iota^*F(H)\iota
 =\frac1\pi\int_{\mathbb R^2}
    \overline\partial\widetilde F_M(s+it)
    \bigl(\widehat R(s+it)-\iota^*R(s+it)\iota\bigr)\,ds\,dt.
\]
For cells outside the strip, write
\[
 r_1=|x-y|,
 \qquad
 r_2=x_2.
\]
Then \eqref{eq:far-resolvent-difference} and
\[
 e^{-c|t|(r_1+r_2)}
 \le
 C_m|t|^{-2m}\br{r_1}^{-m}\br{r_2}^{-m}
 \quad(0<|t|\le1)
\]
give, when $M>N'+2m-1$,
\[
 \begin{split}
 &\|\widehat\chi_x(F(\widehat H)-\iota^*F(H)\iota)
                  \widehat\chi_y\|_1\le C_{F,m}\br{x-y}^{-m}\br{x_2}^{-m}
             \int_{-1}^1|t|^{M-N'-2m}\,dt .
 \end{split}
\]
Since
\[
 M-N'-2m>-1,
 \qquad
 \int_{-1}^1|t|^{M-N'-2m}\,dt
 =\frac{2}{M-N'-2m+1},
\]
we obtain
\[
 \|\widehat\chi_x(F(\widehat H)-\iota^*F(H)\iota)
                  \widehat\chi_y\|_1
 \le C'_{F,m}\br{x-y}^{-m}\br{x_2}^{-m}.
\]
The cells meeting the strip were treated in
\eqref{eq:smooth-offdiag-trace}.  Since the constant part of $f$ cancels under compression,
\[
 f(\widehat H)-\iota^*f(H)\iota
 =
 F(\widehat H)-\iota^*F(H)\iota,
\]
and the estimates above give
\[
 \|\widehat\chi_x
 (f(\widehat H)-\iota^*f(H)\iota)
 \widehat\chi_y\|_1
 \le C_{f,m}\br{x-y}^{-m}\br{x_2}^{-m}.
\]
This is \eqref{eq:CEdge}, with $\nu_f=0$.
\end{proof}

Suppose an open interval $I_0\supset\overline\Delta$ lies in the region
of complete localization, as defined in \cite[Section~2]{GKchar}.
The quantitative consequences used here are
\cite[Theorem~3 and Eq.~(6.1)]{GKchar}.  For $J\Subset J'\Subset I_0$
and $0<\zeta<1$, they give
\[
 \begin{split}
 \mathbb E\!\left[
  \sup_{\substack{|h|\le1\\\operatorname{supp}h\subset\overline J}}
  \|\chi_xh(H_\omega)\chi_y\|_2^2\right]
 &\le C_{J,\zeta}e^{-c|x-y|^\zeta}\text{ and }
 \mathbb E\!\left[
  \sup_{E\in J}\|\chi_x\mathbf1_{(-\infty,E]}(H_\omega)\chi_y\|_2^2
 \right]
 \le C_{J,\zeta}e^{-c|x-y|^\zeta}.
 \end{split}
\]
The strict Fermi projection obeys the same bound, by taking
$E'\uparrow E$ inside $J'$ and using Fatou's lemma.
For $J=(a,b)$ and $f\in\mathcal B_1(J)$, define
\[
 h_f(t)=f(t)-f_+-(f_--f_+)\mathbf1_{(-\infty,a)}(t).
\]
Then $\operatorname{supp}h_f\subset\overline J$, $\|h_f\|_\infty\le2$,
and
\[
 f(H_\omega)-f_+
 =(f_--f_+)\mathbf1_{(-\infty,a)}(H_\omega)+h_f(H_\omega).
\]
Cauchy--Schwarz in probability gives
\[
 \mathbb E F_J(x,y;\omega)\le C_Je^{-c_J|x-y|^\zeta},\qquad
 F_J(x,y;\omega):=
 \sup_{f\in\mathcal B_1(J)}
 \|\chi_x(f(H_\omega)-f_+)\chi_y\|_2.
\]
For $p>2$ and integer $m\ge1$, set
\[
 Z_{J,m}=\sum_{x,y}\br x^{-p}\br{x-y}^mF_J(x,y;\omega).
\]
As in the discrete case,
\[
 \mathbb E Z_{J,m}\le C_J
 \left(\sum_x\br x^{-p}\right)
 \left(\sum_v\br v^m e^{-c_J|v|^\zeta}\right)<\infty,
\]
and therefore
\[
 F_J(x,y;\omega)\le Z_{J,m}(\omega)\br x^p\br{x-y}^{-m}
 \quad\text{on }\{Z_{J,m}<\infty\}.
\]
This is \COne with the fixed exponent $p$.

For the localized eigenbasis, use
\cite[Corollary~1 and Corollary~3(ii)--(iii), Eqs.~(3.13)--(3.17)]{GKchar}.
For normalized eigenfunctions these estimates imply, with fixed
$\nu\ge0$,
\[
 \|\chi_x\psi_n\|
 \le C_{J,\omega}\br{x_n}^{\nu}e^{-c|x-x_n|^\zeta}
 \le C_{J,m,\omega}\br{x_n}^{\nu}\br{x-x_n}^{-m}
 \quad(m\ge1).
\]
Let $\Omega_{{\rm basis},J}$ be the probability-one event on which this
basis and finite multiplicity hold, and let
$\mathcal J_{\mathbb Q}(I_0)$ denote the countable set of intervals with
rational endpoints and closure in $I_0$.  A single probability-one event is
\[
 \Omega_*=
 \bigcap_{J\in\mathcal J_{\mathbb Q}(I_0)}\Omega_{{\rm basis},J}
 \ \cap\!
 \bigcap_{\substack{J\in\mathcal J_{\mathbb Q}(I_0)\\m\in\mathbb N}}
 \{Z_{J,m}<\infty\},\qquad \mathbb P(\Omega_*)=1.
\]  Inclusion in a larger rational interval
gives \COne--\CTwo for every $J\Subset\Delta$ on $\Omega_*$.

For $b_0\in\mathbb R$, choose $\kappa>0$ so that
$H_\omega+\kappa\ge1$ and $b_0+\kappa>0$.  The scalar functional calculus
yields
\[
 \mathbf1_{(-\infty,b_0]}(H_\omega)
 \le(b_0+\kappa)^2(H_\omega+\kappa)^{-2}.
\]
With $R_\kappa=(H_\omega+\kappa)^{-1}$,
\Cref{lem:local-schatten} gives
\[
 \begin{split}
 \Tr\bigl(\chi_x\mathbf1_{(-\infty,b_0]}(H_\omega)\chi_x\bigr)
 &\le(b_0+\kappa)^2\|R_\kappa\chi_x\|_2^2=(b_0+\kappa)^2\sum_z\|\chi_zR_\kappa\chi_x\|_2^2\le C_{b_0}.
 \end{split}
\]
This proves \CThree, uniformly in $x,\omega$.
\Cref{lem:dirichlet-boundary} proves \CFour.

Thus \COne--\CFour hold almost surely throughout $\Delta$ on one
probability-one event. The localization input is the complete-localization
theory of Germinet--Klein \cite{GKchar}; quantization and mobility-gap
Hall-index results for random Landau Hamiltonians are developed in
\cite{GKS2,Taarabt}. Applying \Cref{thm:main} proves \Cref{cor:AL}.

\end{document}